\documentclass[journal]{IEEEtran}

\usepackage{graphicx}
\usepackage{mathptmx}
\usepackage{amsmath}
\usepackage{amssymb}
\usepackage{amsthm}

\newtheorem{definition}{Definition}

\newtheorem{lemma}{Lemma}

\usepackage{amsfonts}
\newcommand{\Z}{\mathbb{Z}}

\usepackage[ruled,norelsize]{algorithm2e}

\usepackage{rotating} 
\usepackage{slashbox}
\usepackage{multirow}
\usepackage{breqn}

{}
{}
\newtheorem{claim}{Claim}{}

\theoremstyle{definition}

\def\N{\mathbb N}
\def\Z{\mathbb Z}

\usepackage{booktabs}  

\usepackage[T1]{fontenc}
\usepackage{amsthm}

\newtheorem{theorem}{Theorem}

\usepackage{mathtools}
\usepackage{tikz}

\DeclarePairedDelimiter\floor{\lfloor}{\rfloor}

\begin{document}

\title{AMOUN: Asymmetric lightweight cryptographic scheme for wireless group communication}

\author{Ahmad~Mansour,~
        Khalid~M.~Malik,
        ~and~Niko~Kaso
                 
\thanks{A. Mansour is with Dura Automotive Systems LLC, Auburn hills, MI, 48326 USA (e-mail: aamansour@oakland.edu).}
\thanks{K. M. Malik is with the Department of Computer Science and Engineering, Oakland University, Rochester, MI, 48309 USA (e-mail: mahmood@oakland.edu).}
\thanks{N. Kaso was with the Department of Mathematics and Statistics, Oakland University, Rochester, MI, 48309 USA  (email: kaso@oakland.edu).}
}

\markboth{}%
{Shell \MakeLowercase{\textit{et al.}}: Bare Demo of IEEEtran.cls for IEEE Journals}

\maketitle

\begin{abstract}

Multi-recipient cryptographic schemes provide secure communication, between one sender and multiple recipients, in a multi-party group. Providing secure multi-party communication is very challenging, especially in dynamic networks. Existing multi-recipient cryptographic schemes pose a variety of limitations. These include high computational overhead for both encryption and decryption, additional communication overhead and high setup cost due to change in membership, and collusion among recipients. In order to overcome these limitations, this paper introduces a novel asymmetric multi-recipient cryptographic scheme, $AMOUN$. In the proposed scheme, to better utilize network resources, the sender transmits a ciphertext containing different messages to multiple recipients, where each recipient only allowed to retrieve its own designated message. Security analysis demonstrates that the proposed scheme is indistinguishable under adaptive chosen plaintext attack. Quantitative analysis reveals that lightweight $AMOUN$ shows lower average computational cost than both RSA and Multi-RSA, for both encryption and decryption, even when the key sizes are four times larger. For a given prime size, in case of encryption, $AMOUN$ shows $98\%$ and $99\%$ lower average computational cost than RSA and Multi-RSA, respectively. For decryption, $AMOUN$ shows a performance improvement of $99\%$ compared to RSA and Multi-RSA.

\end{abstract}


\IEEEpeerreviewmaketitle

\section{Introduction}\label{sec:Intro}

Recently, multi-party communication is experiencing a significant surge in research interest due to the massive growth of its applications \cite{shi2014celltv,thomasson2015system, vijayakumar2016effective}. In the presence of adversaries and the possibility of illicit cooperation between two or more group members, also known as collusion, and in the absence of pre-established trust for group keys in highly dynamic networks, securing multi-party communication becomes very challenging. In multi-party communication, multi-party group, also known as a receiving group, consists of nodes/users having the same demands at a specific time and location \cite{mahmood2009autonomous}. The membership of group may continuously vary which is caused by the change in group members' preferences, characteristics of wireless media, and nodes' mobility \cite{mahmood2009autonomous,mansour2019amoun}. Multi-party secure communication (MSC), also known as secure one-to-many group communication, is used by an individual sender to share information with a group of recipients. In MSC, the single encrypted sent message could have the same or different contents for a set of recipients. 

Multi-recipient cryptography, a key enabler for MSC, is required to set up secure channels for the data exchange in different types of end-user applications \cite{shi2014celltv,thomasson2015system, vijayakumar2016effective}. It is also needed for securing key management including key establishment and key exchange between sender and receiving group in a highly dynamic networks, such as vehicular ad hoc networks (VANETs) and wireless sensor networks (WSNs) \cite{cheikhrouhou2016secure, vijayakumar2016dual}. Since standard encryption schemes do not allow to exploit batching for MSC, as each recipient receives a separate encrypted message; therefore, it results in a degradation in performance and inefficient use of bandwidth. Thus, there is a need to have multi-recipient cryptographic schemes, where all recipients in a group either require same or different contents from an individual sender in a single message \cite{mansour2019amoun,bellare2007multi}. Existing multi-recipient cryptographic schemes mainly focused on securing one-to-many communications without considering the computational and bandwidth constraints, which makes such schemes less practical for dynamic networks \cite{bellare2007multi,mansour2017multi}.

There has not been much research carried out on schemes aiming to share different contents within a single ciphertext, by a sender to multiple recipients. However, many applications can be benefited from such schemes. For example, in VANETs, Road Side Unit (RSU) can share multiple secret keys to member of multiple sub-groups in single ciphertext using multi-recipient cryptography. To accomplish this, RSU sends the ciphertext to cluster heads, and then each cluster head multicasts the secret key among its group members \cite{daeinabi2014advanced,gazdar2011secure}. Similarly, on-the-fly firmware can be updated to receiving nodes, in dynamic networks such as VANETs or WSNs \cite{hasrouny2017vanet,engoulou2014vanet}. In multicast publish/subscribe-based applications, the sending or receiving nodes could be a member of multiple groups/sub-groups \cite{mahmood2009autonomous,mahmood2009autonomous2,mahmood2008autonomous3}. For such applications, the sender can send the information to different groups/sub-groups in a single ciphertext, and each recipient can retrieve its corresponding group message by decrypting the ciphertext. Other futuristic applications include scenarios where multi-modal data of specific entity needs to be sent remotely to multiple entities located in same physical premises, and every receiving entity has different access control rights to get a portion of the multi-modal data. To fulfill the requirements of such applications, this paper proposes an efficient asymmetric multi-recipient cryptographic scheme, $AMOUN$, which is provably indistinguishable under adaptive chosen plaintext attack ($IND-CPA$). $AMOUN$ enables a sender to send different information to multiple recipients in a single attempt to save computational and bandwidth resources. This research focuses on the confidentiality of transmitted data between the sender and receiving group; therefore, we assume that there exists already a reliable authentication protocol \cite{hasrouny2017vanet,engoulou2014vanet,pathan2016security}.

Existing multi-recipient cryptographic solutions introduce, a) the computational and communication overhead due to change in group membership, b) the ciphertexts concatenation constraints, c) the need for key distribution, d) possibility of compromising group privacy, e) the need for group key, f) collusion among recipients, and g) the limitation of sending the same message for all recipients. To overcome the above-mentioned challenges, $AMOUN$ effectively leverages the mathematical formulations of Chinese Remainder Theorem (CRT) \cite{ding1996chinese}, prime factorization \cite{rivest1978method}, discrete logarithm \cite{diffie1976new}, and the use of noise parameter \cite{coron2011fully,gentry2009fully}. Furthermore, the proposed scheme avoids the requirement of additional computational and communication overhead when group membership of receiving group changes. It only requires receiving members to keep their private keys safe and share their public keys. Moreover, it neither requires a group key nor a sender to know the private keys of recipients to generate the public key. Also, similar to RSA, $AMOUN$ requires minimal collusion between the sender and receiving group and allows for plausible deniability \cite{mansour2017multi,bassous2015ambiguous}. Additionally, $AMOUN$ does not suffer from the key distribution problem due to its fully asymmetric nature. Therefore, $AMOUN$ is more practical because there is no need to either concatenate ciphertexts or to share secrets between members of receiving group. On top of that, there is a significant decrease in terms of computational overhead compared to other existing asymmetric schemes.

In \cite{mansour2019amoun}, we introduced the basic idea of $AMOUN$ for MSC. This paper extends the basic idea to a complete multi-recipient asymmetric cryptographic scheme. The major contributions of this paper are as follows: First, it presents $AMOUN$'s algorithms along with formal security proof to verify that the proposed scheme is indistinguishable under adaptive chosen plaintext attack ($IND-CPA$). Second, to prove $AMOUN$'s applicability, it presents a validation analysis by giving a mathematical proof. Third, it presents a detailed empirical analysis to prove $AMOUN$'s effectiveness in terms of its low computational and communication overhead.

The rest of this paper is structured as follows: Section \ref{sec:RelatedWork} reviews the related work and their shortcomings. Section \ref{sec:Preliminaries} presents the mathematical notation and definitions. Section \ref{subsec:Algorithms} discusses $AMOUN$'s algorithms. Section \ref{sec:ValidationAnalysis} explains the validation analysis of proposed scheme. Section \ref{sec:SecurityAnalysis} illustrates the security analysis of our solution and possible attacks. Section \ref{sec:TimeComplexity} provides the time complexity analysis. Section \ref{sec:Evaluation} shows the evaluation of proposed scheme. Section \ref{sec:Discussion} discusses the comparative and quantitative analysis of $AMOUN$ with other asymmetric cryptographic schemes. Lastly, section \ref{sec:Conclusion} concludes this paper.

\section{Related Work} \label{sec:RelatedWork}

This section reviews the recent state of the art one-to-many cryptographic schemes and provides a brief overview of their shortcomings. The related works of one-to-many cryptography can be divided into two main categories: symmetric cryptographic and asymmetric cryptographic schemes. Table \ref{tbl:RW} presents comparison of $AMOUN$ with existing one-to-many cryptographic schemes considering following factors: concatenation of ciphertexts, need for key distribution, communication and computational overhead between sender and recipients for change in group membership, computational overhead for encryption and decryption, additional recipient setup costs, collusion among recipients, need for group key, and threat of group privacy.


\begin{sidewaystable} 
	\centering
	\caption{Comparison of AMOUN with existing one-to-many cryptographic schemes}
	\label{tbl:RW}
	
	\begin{tabular}{p{9.2cm}lllllllll}
		\hline
		\backslashbox[5.45cm]{\textbf{Factors}}{\textbf{Scheme}}& 
		\textbf{\cite{bassous2015ambiguous}} & 
		\textbf{\cite{canetti1999efficient}} &
		\textbf{\cite{perrig2005tesla}}&
		\textbf{\cite{dodis2002public}}&
		\textbf{\cite{nguyen2005rsa}}& 
		\textbf{\cite{rivest1978method}}&
		\textbf{\cite{mansour2017multi}}&	
		\textbf{$AMOUN$} \\ \hline
		\textbf{Type}                                                             &Symme &Symme &Symme &Asymm&Asymm&Asymm&Asymm&Asymm\\ 
		\textbf{Concatenation of ciphertexts}                                     &No &No &No  &No &No &Yes&No&No  \\ 
		\textbf{Need for key distribution}                                        &Yes&Yes&Yes&No &No &No &No&No  \\ 
		\textbf{Computational overhead for encryption and decryption}             &Low &Low &Low &High &High&High &High &Low  \\ 
		\textbf{Communication overhead for change in group membership}            &Yes&Yes&Yes&Yes&Yes&Yes&No&No  \\ 
		\textbf{Setup cost due to change in membership}                           &No &Yes&Yes&Yes&Yes&No &No&No  \\ 
		\textbf{Collusion among recipients}                                        &No &Yes&No  &Yes&Yes&No &No&No  \\ 
		\textbf{Need for group key}                                               &No &Yes&No  &Yes&No &No &No&No  \\ 
		\textbf{Threat of group privacy}                                          &No &Yes&No  &Yes&Yes&No &No&No  \\  \hline                          
		
	\end{tabular}	
\end{sidewaystable} 


\subsection{Symmetric Cryptographic Schemes} 

Symmetric cryptographic schemes for one-to-many systems suffer from many disadvantages. In \cite{micciancio2006corrupting}, Micciancio et al. provide a competitive analysis of different types of one-to-many cryptographic schemes and group key distribution schemes. They explain that such algorithms encounter the group privacy problem, where a subset of users needs to collude with each other to combine their secret information to decrypt the transmission of the sender. This study shows that all of the studied one-to-many cryptographic protocols require group key distribution. Unlike these symmetric schemes, $AMOUN$ doesn't have any such requirement, as parties only need to publish their public keys to a common source for the sender to access them. Another effort using a symmetric scheme can be found in \cite{canetti1999efficient}. The authors discussed minimal storage and minimal communication based cryptographic scheme. In the minimal storage method, every recipient carries private and group keys that they use when a change in the recipients' group takes place. On the other hand, their minimal communication method uses a tree structure with the sender at the root and recipients at leaf nodes. Every recipient knows the keys of all nodes between the root and themselves. A change in the receiving group trickles down to all the nodes on the path between the root and the modified node. Engaging other parties in this communication setup and teardown eventually creates privacy concerns about the shared information between the nodes on the same path in the tree. In contrast to \cite{canetti1999efficient}, $AMOUN$ does not require additional communication and computational costs at recipients side when a change in receiving group occurs. Moreover, our scheme is a native asymmetric implementation where all recipients are independent with their public and private key pair, and no centralized key distribution is needed.

Ambiguous Multi-Symmetric Cryptography (AMSC) is another one-to-many symmetric cryptographic system based on CRT \cite{bassous2015ambiguous}. It narrates that the knowledge of all moduli in CRT is fundamental to get all messages from the ciphertext. While AMSC has been proven to be faster than many other symmetric cryptographic schemes, it requires collusion between sender and recipients for exchanging keys \cite{bassous2015ambiguous, schneier2007applied}. Also, it faces the key distribution problems, one of the fundamental challenges with symmetric cryptography \cite{rivest1978method}. These can be solved using asymmetric key distribution protocols. Managing symmetric algorithms requires resources not needed in an asymmetric implementation. Considering a system with limited resources, key distribution can become impossible or extremely difficult as the number of participants increases. In addition, such an expansion leads to an exponential increase in memory requirements for key storage, which is impossible in edge-connected devices. Due to its asymmetric nature, $AMOUN$ does not require an additional layer of key distribution prior to the secure communication.

Few schemes were proposed that integrate authentication and confidentiality using symmetric one-to-many cryptographic techniques \cite{perrig2001efficient,zhao2012survey}. For example, a well-known approach in this regard\cite{perrig2005tesla} proposes one-way hash chains and loosely synchronization protocols. However, this scheme has many limitations. First, computational costs and communication delays are introduced by the time needed to authenticate the messages. This impact is further magnified due to the need for a periodic key regeneration, which is an additional overhead, particularly in the resource-constrained networks. Also, this scheme is not scalable and suffers from key distribution issues when it assumes that one-way hash chains are initially sent via a secure channel. These drawbacks are solved in our schemes, where there is no need for the periodic key regeneration. In addition, if there is a change in the receiving group, $AMOUN$ requires no communication overhead. This makes $AMOUN$ ideal for systems, such as VANETs, where there exists little or no trust among individual recipients. Lastly, as $AMOUN$ solely relies on the key generation by the system's users, it involves reduced setup cost compared to existing symmetric one-to-many cryptographic solutions.

\subsection{Asymmetric Cryptographic Schemes} 

To solve the above-mentioned challenges of symmetric one-to-many cryptographic schemes, Dodis et al. \cite{dodis2002public}, adopted the minimal communication based cryptographic scheme defined in \cite{canetti1999efficient} using asymmetric cryptographic. However, this solution introduces a huge communication overhead for revoking or adding recipients. In $AMOUN$, the sender simply changes receiving group members without ever having to notify other receiving group members. Additionally, with $AMOUN$, there is much less setup cost, since we do not have to construct whole tree structure and generate keys ahead of time, which adds an extra computational overhead as it uses asymmetric cryptographic.

On the other hand, RSA is a standard cryptographic scheme used today, created by Ron Rivest, Adi Shamir, and Leonard Adleman \cite{rivest1978method}. RSA and $AMOUN$ have a similar mathematical base as they both heavily rely on properties of modulus \cite{ding1996chinese} and prime factorization problem \cite{rivest1978method}. Security of RSA depends on the difficulty of factoring two large prime numbers and the property of two modular inverses in a Euler's phi space of a number. These two characteristics make the RSA scheme extremely hard to break, for large enough key size, without knowing the private key \cite{rivest1978method}. In normal asymmetric cryptographic systems, including RSA, a one-to-many cryptographic system could be implemented by concatenating ciphertexts together and sending the resultant message. These systems face several challenges. First, both the sender and recipient have to agree on a ciphertext offset in the concatenation. To accomplish this, the sender and recipient will need to collude before the start of the multi-party communication. This violates the principles of asymmetric cryptography that abandon the need for the parties to communicate prior for sending the information. Also, manipulating the offsets for every message introduces the overhead of an extra communication per message, while a constant offset requires padding, for example addition of zeros, to account for empty spaces leading to an inefficient larger ciphertext size. Another limitation is that ciphertexts are only concatenated but not mixed. This enables attackers to target specific ciphertext or listen to only parts of the entire concatenation to find specific messages. In addition, if the attacker knows the size of ciphertexts, he can determine the number of messages exchanged. Hence, we can conclude that such concatenation compromises the security of a message exchange \cite{mansour2017multi}.

Since the creation of RSA, many researchers have proposed new derivative cryptographic schemes using it as a base. One of the most known modifications is CRT-RSA \cite{quisquater1982fast}. This scheme uses CRT to split up the decryption key into two pieces to make the exponentiation faster. Although, CRT-RSA uses CRT, however, it is not designed for secure group communication \cite{quisquater1982fast}. Shared RSA, another extension of RSA, can be used for secure communication from one-to-many, or many-to-many \cite{nguyen2005rsa}. The major problem with Shared RSA is that it requires collusion among recipients in order to work successfully. This means that each recipient must be individually trusted. More importantly, in this scheme, all recipients need to share a common secret among themselves. This not only adds additional recipient initial setup cost, but also introduces system communication overhead in case of any change in membership of receiving group. This overhead further increases, as the number of the recipients increases. In our scheme, we solved this issue by only using the public keys without the need for membership re-configuration process, which introduces huge computational and communication overhead.

In \cite{mansour2017multi}, authors proposed a one-to-many asymmetric cryptographic scheme, Multi-RSA, which is based on RSA and CRT to address the limitations of existing solutions for one-to-many asymmetric RSA schemes. In this paper, we compare our scheme with Multi-RSA since, like RSA, it relies on properties of modulus and prime factorization problem \cite{ding1996chinese,rivest1978method}. Although, Multi-RSA has the same communication overhead as $AMOUN$; however, it suffers heavily from using RSA cryptographic scheme as a base, since RSA scheme is known to have high computational cost for both encryption and decryption, which makes it hard to implement Multi-RSA in dynamic mobile network environments and makes this scheme less scalable. $AMOUN$ achieves all of the benefits of Multi-RSA with significant improvement in terms of time required for both encryption and decryption by introducing a lightweight one-to-many asymmetric cryptographic scheme.

\section{Preliminaries} \label{sec:Preliminaries}

This section introduces the notation used by $AMOUN$'s algorithms and provides formal definitions of the proposed scheme. 

\subsection{Mathematical Notation}

Let $\N=\{1,2,3,\ldots,\}$ be the set of of natural numbers and $\Z=\{\ldots,-3,-2,-1,0,1,2,3,\ldots\}$ be the set of integer numbers. We write $a \equiv b~(\bmod~n)$ if $a$ and $b$ are two integers that have the same remainder when divided by $n \in \N$. For each $n \in \N$, $\Z_n=\{\overline0, \overline1, \overline2,\ldots,\overline{n-1}\}$ is the ring of integers modulo $n$. Elements of $\Z_n$ are called classes of the set of integers modulo $n$, and each class contains all integers that have the same residue modulo $n$. For instance, $a \equiv b~(\bmod~n)$, if and only if $a$ and $b$ appear in the same class of integers modulo $n$, e.g. $\bar{0}=\{0, n, 2\cdot n, \ldots, k\cdot n, \ldots\}$. The length of a string $Y$ is denoted by $|Y|$, whereas the set of binary strings of finite length is denoted by $\{0,1\}^*$. Also, $1^\alpha$ refers to the string $\overbrace{11\ldots 1}^\alpha$. For future reference in the section, we denote $F_1~\xleftarrow{\text{\$}}~F$ if $F_1$ is a set taken uniformly at random from $F$. In the case when there are several sets $F_1, F_2, \ldots, F_v$ obtained uniformly at random from the set $F$, we will use the abbreviation $F_1, F_2, \ldots, F_v ~\xleftarrow{\text{\$}}~F$. For every $v \in \N$ we also use $F~\xleftarrow{\text{\$}}~A(F_1, F_2, \ldots, F_v)$ to mean that $F$ is the output of the randomized algorithm $A$ on inputs $F_1, F_2, \ldots, F_v$. Notation $X~\xleftarrow{\text{\$}}~Y$ is used for the operation of assigning the value $Y$ to $X$.

\subsection{Definitions and Theorems} \label{sec:Def}
Below, we give some definitions and theorems for the $AMOUN$ cryptosystem.

\begin{lemma}
	Let $a, b \in \Z$ and $n \in \N$. We say that
	$a \equiv b~(\bmod~n)$ ($a$ is congruent to $b$ modulo $n$) if and only if $n|(a-b)$ \cite{lindell2014introduction}.
	
\end{lemma}

\begin{theorem}
	\textit{For every $a, b, c, d \in \Z$ and every $n \in \N$ the following statements hold true:}
	
	\begin{itemize}
		\item $a \equiv a~(\bmod~n)$,
		\item $a \equiv b~(\bmod~n)$ implies $b \equiv a~(\bmod~n)$,
		\item $a \equiv b~(\bmod~n)$ and $b \equiv c~(\bmod~n)$ implies $a \equiv c~(\bmod~n)$,
		\item $a \equiv b~(\bmod~n)$ and $c \equiv d~(\bmod~n)$ implies $(a \pm c) \equiv (b \pm d)~(\bmod~ n)$,
		\item $a \equiv b~(\bmod~n)$ and $c \equiv d~(\bmod~n)$ implies $a \cdot c \equiv b\cdot d~(\bmod~n)$,
		\item $a \equiv b~(\bmod~n)$ implies $a^i \equiv b^i~(\bmod~n)$, for every $i \in \N$.\\
	\end{itemize}
\end{theorem}

\begin{definition}
	Let $a, b \in \Z$ be two integers. An integer $d$ is the greatest common divisor of $a$, $b$ if it is a multiple of every common divisor of $a$ and $b$ i.e., if $d_1|a, d_1|b$, then $d_1|d$.
\end{definition}

\begin{theorem}
	
	\textit{Let $n_1, n_2, \ldots, n_z$ be $z$ pairwise relatively prime numbers and $a_1, a_2, \ldots, a_z \in \Z$. Then, the system of equations \label{CRTTT}}

	\[
	\left\{
	\begin{array}{ll}
	x \equiv a_1 &(\bmod~n_1) \\
	x \equiv a_2 &(\bmod~n_2) \\
	&\vdots\\
	x \equiv a_i &(\bmod~n_i)\\
	&\vdots\\
	x \equiv a_z &(\bmod~n_z)\\
	\end{array}
	\right.
	\]
	
\end{theorem}	
has a unique solution modulo
$n=n_1 \cdot n_2 \cdot \ldots \cdot n_z$ \cite{ding1996chinese}.\\\\

\begin{definition}
	A function
	$f: \N~\xrightarrow~[0, 1]$ is called \textsl{negligible}
	if $$\lim_{n\to\infty} f(n)=0$$ and
	$$\lim_{n\to\infty} \frac{f(n)}{\frac{1}{p(n)}}=0$$\\ for every
	polynomial $p:\N~\xrightarrow~\N$ and all $n \ge n_p$, where $n_p$ is
	some natural number.\end{definition}

We now continue by describing the syntax of an encryption scheme. In this research, we define our multi-recipient asymmetric cryptographic scheme, $AMOUN$, by following the standard definitions of asymmetric cryptographic schemes from \cite{bellare2007multi,bellare2000public,cramer1998practical,cramer2003design}. Moreover, we are extending the standard definitions in order to add an $initialization$ algorithm. Therefore, our multi-recipient asymmetric cryptographic scheme $AMOUN = (KG, I, E, D)$ consists of four algorithms:\\

\begin{itemize}
	\item The $key~generation$ algorithm, $KG$, is a probabilistic algorithm that takes in $1^\alpha$ as an input, where $\alpha\in N$ is a security parameter, and returns public and private keys, $K_i^+$ and $K_i^-$, respectively.\\
	
	\item The $initialization$ algorithm, $I$, is a probabilistic algorithm that takes in a list of public keys $K^+_1, ..., K^+_n$ to produce a list of initialization parameters $L =\{\{N'_1, AX_1\}\}, \ldots, \{N'_n, AX_n\}\}$, which is needed for encryption.\\

	\item The $encryption$ algorithm, $E$, is a probabilistic algorithm that takes in the list of parameters $L$, a list of public key elements $e_1, ..., e_n$, a list of coins $r=\{r_1, ..., r_n\}$, and a message vector $M =\{m_1, \ldots, m_n\}$, to produce the ciphertext $C$, where $r$ is provided from the sender to randomize the ciphertext $C$.\\

	\item The $decryption$ algorithm, $D$, is a deterministic algorithm that takes in the private key $K_i^-$ and the ciphertext $C$ to produce either the message $m_i\in \{0, 1\}^*$ or a special symbol $\perp$ to indicate that the ciphertext $C$ was invalid.\\
	
\end{itemize}

Each message $m_i$ in the vector $M$ has been created from a $message~space~MsgSp(v)$, where $1<v<1^\alpha$. Moreover, it is important to clarify that the notation $C~\xleftarrow{\text{\$}}~E_{L}(M)$ is the shorthand for $r~\xleftarrow{\text{\$}}~Coin_{E}(1^\alpha)$; $C~\xleftarrow{\text{\$}}~E_{L}(M, r)$, where $Coin_{E}(1^\alpha)$ generates the random coins list $r$ for the encryption algorithm $E$. Based on the above mentioned discussion, the following experiment is required to return $1$ with probability $1$:

\begin{table}[!h]
	
	\centering
	\label{tbl:Exp}
	
	\begin{tabular}{ll}
		&For $i=1, ..., n$ do ($K^+_i, K^-_i$)$~\xleftarrow{\text{\$}}~KG(1^\alpha)$ EndFor; \\ 
		&For $i=1, ..., n$ do $L_i~\xleftarrow{\text{\$}}~I(K^+_i)$ EndFor; \\ 
		&$M~\xleftarrow{\text{\$}}~MsgSp(v)$; $C~\xleftarrow{\text{\$}}~E_{L}(M)$; \\ 
		&$j~\xleftarrow{\text{\$}}~\{1, ..., n\};$ If $D_{K_j^-}(C)=m_i$ then return $1$ else return $0$;         \\
	\end{tabular}
\end{table}

In this research, we will prove that the proposed multi-recipient asymmetric cryptographic scheme, $AMOUN$, is secure against adaptive chosen plaintext attack. We identify a concrete-security version of the standard notion of security of multi-recipient asymmetric cryptographic schemes in the sense of indistinguishability as in \cite{bellare2007multi,bellare2000public}.

\theoremstyle{definition}
\begin{definition} $\textbf{[IND-CPA]}$
	Let $AMOUN = (KG, I, E, D)$ be a multi-recipient asymmetric cryptographic scheme. Let $\mathcal{A}_{cpa}$ be an adversary which runs in two stages, $find$ and $guess$, and has access to an encryption oracle. For $b \in \{0, 1\}$, lets define the following experiment:
\end{definition}

\begin{table}[!h]
	
	\centering
	\label{tbl:ExpCPA}
	
	\begin{tabular}{ll}
		\multicolumn{2}{l}{\textbf{Experiment Exp}$_{AMOUN,~\mathcal{A}_{cpa}}^{cpa-b}$($\alpha$)} \\ 
		&For $i=1, ..., n$ do ($K^+_i, K^-_i$)$~\xleftarrow{\text{\$}}~KG(1^\alpha)$ EndFor; \\ 
		&($M_0, M_1, st$)$~\xleftarrow{\text{\$}}~\mathcal{A}_{cpa}(find, K_1^+, ..., K_n^+)$; \\ 
		&For $i=1, ..., n$ do $L_i~\xleftarrow{\text{\$}}~I(K^+_i)$ EndFor; \\ 
		&$C~\xleftarrow{\text{\$}}~E_{L}(M_b)$; \\ 
		&$d ~\xleftarrow{\text{\$}}~\mathcal{A}_{cpa}(guess, C, st)$; \\ 
		&Return $d$;      \\ 
	\end{tabular}
\end{table}

In the experiment above, the adversary $\mathcal{A}$ performs two stages. In the first stage, $find$, the adversary $\mathcal{A}$ takes in $K^+_1, ..., K^+_n$ and returns two message vectors $M_0$ and $M_1$ of size $n$, where $|M_0|=|M_1|$, which means that the messages in both message vectors have the same size, and $M_0, M_1 \in MsgSp(v)$. In addition to that, the first stage returns the parameter $st$, which indicates the state information that the adversary wants to maintain. In the second stage, $guess$, the encryption oracle provides the adversary $\mathcal{A}$ a challenge ciphertext $C$ which is the encryption of a randomly selected message vector $M_b$, where $b \in \{0, 1\}$. Then, $\mathcal{A}$ should guess which message vector was chosen to get $C$. Therefore, for $IND-CPA$ attack, we define the advantage of the adversary, also known as $\epsilon$, as follows:

\begin{table}[!h]
	
	\centering
	\label{tbl:ExpAdvan}
	
	\begin{tabular}{l}
		\textbf{Adv}$_{AMOUN,~\mathcal{A}_{cpa}}^{cpa}$($\alpha$) = Pr[\textbf{Exp}$_{AMOUN,~\mathcal{A}_{cpa}}^{cpa-0}$($\alpha$) = 0] --\\ Pr[\textbf{Exp}$_{AMOUN,~\mathcal{A}_{cpa}}^{cpa-1}$($\alpha$) = 0].\\
	\end{tabular}
\end{table}

The proposed cryptographic scheme $AMOUN$ is said to be $IND-CPA$ secure if the function \textbf{Adv}$_{AMOUN,~\mathcal{A}_{cpa}}^{cpa}(.)$ is negligible. In other words, $\epsilon$ should be negligible for any random polynomial-time adversary $\mathcal{A}$. In this research, we define the advantage function, $\epsilon$, of the proposed scheme for indistinguishability under adaptive chosen plaintext attack with respect to the $time-complexity$ metric, where the $time-complexity$ is defined as the execution time in the worst-case scenario that the adversary $\mathcal{A}$ requires to perform the above-mentioned experiment, on a specific fixed hardware platform.

\section{Algorithms} \label{subsec:Algorithms}

This section discusses $AMOUN$'s algorithms: key generation, initialization, encryption, and decryption. Both key generation and decryption are performed by the recipients, while initialization and encryption are done by the sender. The main novelty of $AMOUN$ is the ability to effectively integrates the mathematical formulations of CRT \cite{ding1996chinese}, prime factorization problem \cite{rivest1978method}, discrete logarithm problem \cite{diffie1976new}, and the use of the noise parameter \cite{coron2011fully,gentry2009fully}, in order to generate a lightweight multi-recipient asymmetric cryptographic scheme. 

In order for $AMOUN$ to take advantage of CRT, all $N_i$s must be pairwise relatively prime. In this research, as the size of prime numbers that are used is $1024-bits$ or larger, we assume that there is little to no chance of prime collision, where the probability that two recipients pick the same prime is extremely low \cite{mansour2019amoun,lindell2014introduction,barker2019transitioning}. Table \ref{tbl:symbols} shows the glossary of symbols used by $AMOUN$.

\subsection{Key Generation $(KG)$}

This algorithm generates four distinct random prime numbers $k_i$, $p_i$, $q_i$, and $v_i$. These primes are generated by satisfying the following chain of inequalities:

\begin{equation}\label{E1}
1^\alpha>k_i\cong p_i\cong q_i>v_i.
\end{equation}

Next, the recipient generates a positive natural random number, $y_i$, with size less than the size of the prime $v_i$, as shown in equation \eqref{E2}.

\begin{equation}\label{E2}
y_i=Random(1,v_i).
\end{equation}

Since $v_i$ is a prime number and $y_i$ is strictly less than $v_i$, which implies that $gcd(v_i,y_i)=1$. Therefore, by the extended Euclidean algorithm \cite{lindell2014introduction,motzkin1949euclidean}, the recipient is able to obtain the inverse of $y_i$ modulo $v_i$, $y_i^{-1}$. Thus, $\exists y_i^{-1}~(\bmod~v_i)$ such that:

\begin{equation}\label{E3}
y_i^{-1}\cdot y_i \equiv y_i\cdot y_i^{-1} \equiv 1~(\bmod~v_i).
\end{equation}

So far, the recipient has obtained 6 integer numbers, from which $k_i$, $v_i$, and $y_i$ make the private key, $K_i^-<k_i, v_i, y_i>$, needed to decrypt the message. The other three, $p_i, q_i,$ and $y_i^{-1}$, are used to generate the elements of the public key $K_i^+<N_i,e_i,d_i>$, satisfying $gcd(N_i, e_i)=1$, $gcd(N_i, d_i)=1$, and $gcd(e_i, d_i)=1$, as follows:

\begin{equation}\label{E4}
N_i=k_i\cdot p_i,
\end{equation}

\begin{equation}\label{E5}
e_i = k_i\cdot q_i+y_i^{-1},
\end{equation}

\begin{equation}\label{E_d}
d_i = v_i^{k_i}~(\bmod~N_i).
\end{equation}

\begin{table}[!t]
	
	\centering
	\caption{Glossary of symbols}
	\label{tbl:symbols}
	
	\begin{tabular}{ll}
		\toprule
		\textbf{Symbol}   & \textbf{Description}       \\  \midrule 
		$m$               & Plaintext (Message)                                    \\ 
		$C$               & Ciphertext                                             \\ 
		$K^+_i$           & Public key for recipient $i$                           \\ 
		$K^-_i$           & Private key for recipient $i$                          \\ 
		$N_i, e_i, d_i$   & Parameters from public key                             \\ 
		$N'_i, AX_i$      & Parameters computed in Initialization                      \\ 
		$X$               & Product of all $N_i$s                                 \\ 
		$k_i, p_i, q_i, v_i$   & Primes used in Key Generation                          \\ 
		$y_i$             & Integer number used in Key Generation                 \\ 
		$y^{-1}_i$        & Modular inverse of $y_i$                               \\ 
		$t_i, f_i$        & Integer numbers used in Initialization                 \\ 
		$A_i$             & Modular inverse calculated in Initialization          \\ 
		$e'_i, S_i$       & Parameters calculated in Encryption                      \\ 
		$n$               & Number of recipients                                   \\ 
		$\alpha$          & Size of keys in bits                                   \\ 
		$gcd$             & Greatest Common Divisor                                 \\ 
		
		\bottomrule
	\end{tabular}
\end{table}

Algorithm \ref{algo:KeyGen} presents the steps for Key Generation algorithm.

\RestyleAlgo{rules}\LinesNumbered   


\begin{algorithm}
	\SetKwData{Left}{left}
	\SetKwData{This}{this}
	\SetKwData{Up}{up}
	\SetKwFunction{Union}{Union}
	\SetKwInOut{Input}{input}
	\SetKwInOut{Output}{output}
	\caption{Pseudo code of Key Generation ($KG$)}
	\label{algo:KeyGen}

	\Input{Size of keys in bits, $\alpha$, size of array $Keys$, $w$.}
	\Output{Array $Keys$ which contains both $K_i^+$ and $K_i^-$, in addition to elements needed to generate them. \newline}
	
	Initialize $Keys$ to an $array$ of size $w$\;
	
	\For {$i \leftarrow 0$ \KwTo $3$}
	{
		$Keys[i] = GenerateRandomPrime(\alpha)$\;
	}
	
	$Keys[4] = GenerateRandom(\alpha)$\;
	$Keys[5] = ModularInverse(Keys[4])$\;
	$Keys[6] = Keys[0] \cdot Keys[1]$\;
	$Keys[7] = (Keys[0] \cdot Keys[2] + Keys[5]) mod Keys[6]$\;
	$Keys[8] = POW(Keys[3], Keys[0])~mod~Keys[6]$\;
	
	\KwRet {$Keys$}\; 
	
\end{algorithm}


\subsection{Initialization $(I)$}

In this algorithm, we describe the setup that the sender needs before encrypting the messages. The sender picks a group of $n>1$ recipients and gets their respective public keys $K^+_1, ...,K^+_n$. Then, the sender computes $N'_i$ for every recipient $i$, where $f_i$ and $t_i$ are random integers less than $1^\alpha$, as shown in equation \eqref{E_N}.

\begin{equation}\label{E_N}
N'_i=N_i\cdot f_i+d_i\cdot t_i,
\end{equation}

Afterwards, the sender computes the product of all $N_i$s, denoted as $X$, as shown in equation \eqref{E6}. Since $N_i$s are pairwise relatively prime, then by the extended Euclidean algorithm \cite{lindell2014introduction,motzkin1949euclidean}, we have that $gcd(N_i,\frac{X}{N_i})=1$, $\forall i$, which implies that there exists a modular inverse $A_i$, as shown in equation \eqref{E7}. Using $A_i$ and $\frac{X}{N_i}$, the sender computes $AX_i$ for each recipient $i$, as shown in equation \eqref{E8}. Next, the sender produces a list of initialization parameters $L =\{\{N'_1, AX_1\}\}, \ldots, \{N'_n, AX_n\}\}$, which is needed for encryption phase.

\begin{equation}\label{E6}
X= \prod_{i=1}^{n} N_i,
\end{equation}

\begin{equation}\label{E7}
A_i \cdot \frac{X}{N_i} \equiv 1~(\bmod~N_i), 
\end{equation}

\begin{equation}\label{E8}
AX_i= A_i \cdot \frac{X}{N_i}.
\end{equation}

Moreover, the sender computes the initialization parameters list $L$ only once for a specific receiving group. If any change in membership, due to the addition or removal of members, occurs in the receiving group, the sender applies the following modifications with minimum overhead to the list of recipients: a) adding $L_i$ to the list when new recipient $i$ is added to the receiving group, or b) removing $L_i$ from the existing list when the corresponding $i$th recipient leaves the group. In addition to that, there is no need for communication between the sender and recipients when any change happens to the receiving group. Therefore, the challenges due to change in the membership are effectively resolved by $AMOUN$.

Algorithm \ref{alg:Init} presents the steps for Initialization algorithm.


\begin{algorithm}
	\SetKwData{Left}{left}
	\SetKwData{This}{this}
	\SetKwData{Up}{up}
	\SetKwFunction{Union}{Union}
	\SetKwInOut{Input}{input}
	\SetKwInOut{Output}{output}
	\caption{Pseudo code of Initialization ($I$)}
	\label{alg:Init}

	\Input{
		$K^+_1<N_1, e_1, d_1>, ..., K^+_n<N_n, e_n, d_n>$ public keys of recipients.}
	\Output{Array $Init\_L$ where $L_i$ is calculated for every recipient $i$ and $X$ is the last element in the array.\newline}

	Initialize $Init\_L$ to an $array$ of size $n+1$\;
	$Init\_L[n+1] = 1$\;
	Initialize $Init\_N$ to an $array$ of size $n$\;
	Initialize $Init\_A$ to an $array$ of size $n$\;
	Initialize $Init\_AX$ to an $array$ of size $n$\;
	Initialize $f$ to an $array$ of size $n$\;
	Initialize $t$ to an $array$ of size $n$\;
	
	\For {$i \leftarrow 1$ \KwTo $n$}
	{
		$f[i] = GenerateRandom(\alpha)$\;
		$t[i] = GenerateRandom(\alpha)$\;
		$Init\_N[i] = N[i] \cdot f[i] + d[i] \cdot t[i]$\;
		$Init\_L[n+1] \cdot= N[i]$\;
		\tcc*[f]{$Init\_L[n+1]$ contains $X$}	
	}

	\For {$i \leftarrow 1$ \KwTo $n$}
	{
		\tcc*[f]{Finding $A_i$, using the extended Euclidean algorithm, as a modular inverse of $\frac{X}{N_i}$} \\ 
		
		$Init\_A[i]=exEucl((Init\_L[n+1] /N[i]),N[i])$\;
		$Init\_AX[i]= Init\_A[i] \cdot Init\_L[n+1] /N[i]$\;
		$Init\_L [i] = \{Init\_N[i], Init\_AX[i]\}$\;
	}
	
	\KwRet {$Init\_L$}\; 
\end{algorithm}


\texttt{\\}

\subsection{Encryption $(E)$} \label{Enc_sec}

In this algorithm, the sender takes in a message vector $M$ of size $n$ in addition to $L$, where $n>1$. The only constraint for the size of message $m_i$ is that it must be less than $v_i$. Moreover, the sender computes $e'_i$, $\forall i$, using $e_i, N'_i$, and $r_i$, as shown in equation \eqref{E_e}.

\begin{equation}\label{E_e}
e'_i=e_i+N'_i\cdot r_i. 
\end{equation}

The parameter $r_i$ is a different randomly generated coin that is produced using a truly random number generator less than $1^\alpha$. The reason for adding this coin $r_i$ is to randomize the resultant ciphertext in order to avoid IND-CPA attack, as discussed in section \ref{sec:SecurityAnalysis}. Using the parameters $e'_i$ and $AX_i$, the sender computes the encryption parameter $S_i$ for each recipient $i$, as shown in equation \eqref{E8N}.

\begin{equation}\label{E8N}
S_i=e'_i\cdot AX_i.
\end{equation}

After that, using $M$ and $S$, the sender generates the common ciphertext $C$ for all recipients using equation \eqref{E9}. This ciphertext is then sent to the receiving group. Alternatively, messages of multiple recipients can be encrypted in parallel first, and then all resultant sub-ciphertexts, $c_1, \cdots, c_n$, could be added. Also, if memory is a concern for the encryption, $\bmod~X$ can be done on each term individually. So, additions can be split up with modulus running in-between to minimize memory utilization.

\begin{equation}\label{E9}
C=\Big(\sum_{i=1}^{n} m_i\cdot S_i\Big)~(\bmod~X). 
\end{equation}

Algorithm \ref{alg:Enc} presents the steps for Encryption algorithm.


\RestyleAlgo{rules}\LinesNumbered   

\begin{algorithm}
	
	\SetKwData{Left}{left}
	\SetKwData{This}{this}
	\SetKwData{Up}{up}
	\SetKwFunction{Union}{Union}
	\SetKwInOut{Input}{input}
	\SetKwInOut{Output}{output}
	\caption{Pseudo code of Encryption ($E$)}
	\label{alg:Enc}

	\Input{Message vector $M$, list of parameters $L$, and parameter $X$.}
	\Output{Ciphertext $C$.\newline}
	
	Initialize $Enc\_S$ to an $array$ of size $n$\;
	Initialize $Enc\_e$ to an $array$ of size $n$\;
	Initialize $r$ to an $array$ of size $n$\;
	
	$C = 0$\;
	
	\For {$i \leftarrow 1$ \KwTo $n$}
	{
		$r[i] = GenerateRandom(\alpha)$\;
		$Enc\_e[i] = e[i]+Init\_N[i] \cdot r[i]$\;
		$Enc\_S[i] = Enc\_e[i]\cdot Init\_AX[i]$\;
		$C~+= m[i] \cdot Enc\_S[i];$ 	
	}
	$C = C~\bmod~X$\;
	\KwRet {$C$}\;
	
\end{algorithm}


\subsection{Decryption $(D)$}

In this algorithm, each recipient decrypts the ciphertext $C$ using its own private key, $K_i^- <k_i, v_i, y_i>$, in order to retrieve $m_i$, using equation \eqref{E10}.

\begin{equation}\label{E10}
m_i=(C~(\bmod~k_i)\cdot y_i)~(\bmod~v_i).
\end{equation}

Algorithm \ref{alg:Dec} presents the steps for Decryption algorithm.


\begin{algorithm}
	
	\SetKwData{Left}{left}
	\SetKwData{This}{this}
	\SetKwData{Up}{up}
	\SetKwFunction{Union}{Union}
	\SetKwInOut{Input}{input}
	\SetKwInOut{Output}{output}
	\caption{Pseudo code of Decryption ($D$)}
	\label{alg:Dec}

	\Input{Ciphertext $C$, private key $K_i^-<k_i, v_i, y_i>$.}
	\Output{Message $m_i$.\newline}
	
	$m_i = (( C~\bmod~k_i) \cdot y_i )~\bmod~v_i$\;
	\KwRet {$m_i$}\; 
	
\end{algorithm}

\section{Validation Analysis} \label{sec:ValidationAnalysis}

In this section, we validate our proposed cryptographic scheme mathematically by giving an algebraic proof of equation \eqref{E10}. Furthermore, we prove that each recipient $i$ can obtain the message $m_i$ sent to him by applying the private key elements, $k_i, v_i$ and $y_i$ to the ciphertext $C$, as shown in \eqref{E10}. We do so by considering the right hand-side of the equation \eqref{E10}, and through some operations based on congruence properties, we derive the message $m_i$ \cite{lindell2014introduction}. For the proposed scheme to properly decrypt the ciphertext $C$ and obtain the message sent, the relation between the message $m_i$ and the parameters $k_i, v_i, y_i^{-1}, t_i$, and $r_i$, should satisfy the following inequality:

\begin{dmath} \label{New_191}
	m_i < \frac{k_i}{y_i^{-1}+v_i \cdot t_i \cdot r_i}.
\end{dmath}

By substituting $C$ in the right hand-side of \eqref{E10} we get,

\begin{dmath}
	(C~(\bmod~k_i)\cdot y_i)~(\bmod~v_i)= 
\end{dmath}

\begin{dmath} 
	=(\big( (\bigg(\sum_{i=1}^{n} (m_i\cdot e'_i \cdot A_i \cdot \frac{X}{N_i})\bigg)\\ (\bmod~X) )~(\bmod~k_i)\big)\cdot y_i)~(\bmod~v_i). \label{New_1}
\end{dmath}

The following lemma, helps us in the next equality.

\begin{lemma}\label{lemma1}
	Let $a,b,c$ be three integer numbers. If $a|b$, then $(c~(\bmod~b))~(\bmod~a)=c~(\bmod~a)$
\end{lemma}
\begin{proof}
	Let $c=b\cdot q+u$, where $q \in \Z$ and $0 \le u < b$. Since, $a|b$, then
	$b=s\cdot a$ for some $s \in \Z$. So, by plugging this in for
	$c=s\cdot a\cdot b+u$. Hence, $(c~(\bmod~b))~(\bmod~a)\equiv u~(\bmod~a)\equiv c~(\bmod~a)$.
\end{proof}

Since $k_i|N_i=k_i\cdot p_i$ and $N_i | X$, then $k_i | X$ by transitivity property of division. By Lemma \ref{lemma1}, we get that $(C~(\bmod~X))~(\bmod~k_i)\equiv C~(\bmod~k_i)$. 

Thus, equation \eqref{New_1} will become

\begin{dmath} 
	=(\big (\bigg(\sum_{i=1}^{n} (m_i\cdot e'_i \cdot A_i \cdot \frac{X}{N_i})\bigg) ~(\bmod~k_i)\big)\cdot y_i)~(\bmod~v_i). \label{New_2}
\end{dmath}

Because $k_i|\frac{X}{N_j}$, $\forall j \neq i$, $j \in \{1, 2, \ldots, n\}$, it is implied that $m_j\cdot e'_j\cdot A_j \cdot \frac{X}{N_j}\equiv 0~(\bmod~k_i)$. This leads us to the following equation:

\begin{equation}\label{New_3}
=(( \bigg(m_i\cdot e'_i\cdot A_i\cdot \frac{X}{N_i} \bigg)~(\bmod~k_i))\cdot y_i)~(\bmod~v_i).
\end{equation}

Equation \eqref{New_3} contains the modular expression $\big(A_i\cdot \frac{X}{N_i}\big)~(\bmod~k_i)$ which is equal to $\big(\Big(A_i\cdot \frac{X}{N_i}\Big)~(\bmod~N_i)\big)~(\bmod~k_i)\equiv 1$, from Lemma \ref{lemma1} and equation \eqref{E7}. Therefore, $\big(A_i\cdot \frac{X}{N_i}\big)~(\bmod~k_i)$ must also be $1$. Using this, we get 

\begin{equation}\label{New_4}
=(( (m_i\cdot e'_i)~(\bmod~k_i))\cdot y_i)~(\bmod~v_i).
\end{equation}

Substituting $e'_i$ by its equivalent expression given in equation \eqref{E_e}, we obtain 

\begin{equation}\label{New_5}
=(( (m_i\cdot (e_i + N'_i \cdot r_i))~(\bmod~k_i))\cdot y_i)~(\bmod~v_i).
\end{equation}

Furthermore, we simultaneously replace $e_i$ and $N'_i$ with their corresponding equivalent expressions given in \eqref{E5} and \eqref{E_N}, respectively. This leads us to 

\begin{dmath}\label{New_6}
	=(( (m_i\cdot (k_i \cdot q_i + y^{-1}_i + (N_i \cdot f_i + d_i \cdot t_i)\cdot r_i))\\~(\bmod~k_i))\cdot y_i)~(\bmod~v_i).
\end{dmath}

Since $k_i |N_i$ and $k_i | k_i \cdot q_i$, after applying $\bmod~k_i$, we obtain 

\begin{dmath}\label{New_7}
	=(( (m_i\cdot (y^{-1}_i + d_i \cdot t_i \cdot r_i))~(\bmod~k_i))\cdot y_i)~(\bmod~v_i).
\end{dmath}

Based on Fermat's little theorem \cite{weisstein2004fermat}, since $d_i=v_i^{k_i}~\bmod~N_i$, then after applying $\bmod~k_i$ the last equation becomes

\begin{dmath}\label{New_71}
	=(( (m_i\cdot (y^{-1}_i + v_i \cdot t_i \cdot r_i))~(\bmod~k_i))\cdot y_i)~(\bmod~v_i).
\end{dmath}

In order to obtain the message $m_i$, we need $m_i\cdot (y^{-1}_i + v_i \cdot t_i \cdot r_i)$ to be less than $k_i$, as shown in inequality \eqref{New_191}. Therefore, applying $\bmod~k_i$ will not reduce the original expression, $m_i\cdot (y^{-1}_i + v_i \cdot t_i \cdot r_i)$. Thus, eliminating $\bmod~k_i$ we get the following:

\begin{dmath}\label{New_72}
	=( (m_i\cdot (y^{-1}_i + v_i \cdot t_i \cdot r_i))\cdot y_i)~(\bmod~v_i).
\end{dmath}

Because $y^{-1}_i \cdot y_i ~\equiv 1 (~\bmod~v_i)$, by multiplying by $y_i(~\bmod~v_i)$, we get the following equation:

\begin{dmath}\label{New_8}
	=(m_i + m_i \cdot v_i \cdot t_i \cdot r_i \cdot y_i )~(\bmod~v_i).
\end{dmath}

Since the second term in \eqref{New_8} contains $v_i$ and $m_i<v_i$, then after applying $\bmod~v_i$, the equation simplifies to

\begin{equation}\label{New_10}
=m_i.
\end{equation}

Thus, proof of equation \eqref{E10} shows the validity of $AMOUN$.

\section{Security Analysis} \label{sec:SecurityAnalysis}

This section shows how $AMOUN$ is indistinguishable under adaptive chosen plaintext attack ($IND-CPA$) by proving Theorem \ref{MainVIT}. In $IND-CPA$, the adversary can encrypt any number of messages and retrieve the corresponding ciphertext, where every encryption of a message should link to a new ciphertext, even when all $K_i^+$s and $m_i$s do not change \cite{mansour2019amoun,lindell2014introduction}. The right-hand side of equation \eqref{E28} represents the encryption function of our cryptosystem, on which indistinguishability depends. There are three possible distinct cases in terms of knowledge that the adversary might have: a) the first case is when the adversary provides the encryption oracle by $1$ out of $n$ messages, b) the second case when the adversary provides $i$ out of $n$ messages where $1 < i < n$, and c) finally when the adversary provides all the messages to the encryption oracle. For simplicity, we assume that the adversary provides all the messages being encrypted, and the adversary knows members of the receiving group; therefore, the adversary knows what public keys, $K_1^+, ..., K_n^+$, are being used.

\begin{equation}\label{E28}
C=m_1\cdot S_1+...+m_i\cdot S_i+...+m_n\cdot S_n-BX,
\end{equation}

In equation \eqref{E28}, the operand $B$ is the amount of information loss, or the quotient, when the modulus of $C$ is taken over $X$. Therefore, $B$ is dependent on the size of all messages, $f_is, t_is, r_is$, and $X$. On the other hand, to take advantage of CRT in our proposed scheme, we must assume that $n>1$, where $n$ is the number of messages being encrypted.

\begin{theorem}	\label{MainVIT}
	\textit{The proposed cryptosystem, $AMOUN$, is indistinguishable under adaptive chosen plaintext attack, assuming that a) the prime factorization problem is computationally hard, b) the discrete logarithm problem is hard in the group $Z_p$, c) the use of the noise parameter makes the $gcd$ attack impossible, and d) finding the correct point on a given plane and a given line is computationally infeasible.}
\end{theorem}
\begin{proof} We prove Theorem \ref{MainVIT} by proving all the sufficient conditions that are mentioned in the statement of the theorem, following the approach given in \cite{cramer1998practical,cramer2003design}.

	Now, we start by examining the security of the private and public keys, $K_i^-$ and $K_i^+$, respectively. First, we show the security of $K_i^-<k_i, v_i, y_i>$. The element $k_i$ is a prime number, and in order to get $N_i$ of size $2048-bits$, based on the size of $N_i$ that is generally used in RSA \cite{barker2019transitioning}, size of $k_i$ should be $1024-bits$. According to \cite{weisstein2003prime}, there are approximately $2^{1015}$ prime numbers of size less than $1024-bits$, which makes it computationally hard for an adversary to find $k_i$. The same logic as for $k_i$ works to show that finding $v_i$ is computationally infeasible. Since $y_i \in \N$ and $y_i<v_i$, then finding $y_i$ is at least as computationally hard as finding $v_i$. Therefore, choosing a large enough size for $k_i, v_i$, implies the security of $K_i^-$. 
	
	Furthermore, we show the security of $K_i^+<N_i,e_i,d_i>$. The first element of $K_i^+$, $N_i$, is the product of primes $k_i$ and $p_i$, as shown in equation \eqref{E4}. Since the prime factorization of integers is known to be computationally infeasible, for large enough primes, it is computationally infeasible to factor $N_i$ \cite{rivest1978method,lindell2014introduction}.

	\begin{definition}$\textbf{[Primes factoring assumption]}$ For every probabilistic polynomial time adversary $\mathcal{A}$, there is a negligible function $\epsilon$ such that:
		\begin{center}
			Pr[$\mathcal{A}(N_i)\in\{k_i, p_i\}]\le \epsilon(n)$,
		\end{center}
		
		where $k_i$ and $p_i$ are primes with size $n-bits$ and $N_i = k_i \cdot p_i$.
	\end{definition}
	
	According to \cite{lindell2014introduction}, the best-known heuristic asymptotic running time algorithm for prime factorization runs on average in time $2^{O(n^{1/3}\cdot (log~n)^{2/3})}$ to factor a number of size $n-bits$. Therefore, $\epsilon(n)=1/2^{O(n^{1/3}\cdot (log~n)^{2/3})}$. The above complexity implies	the security of Theorem \ref{MainVIT}-(a) and proves the security of $N_i$ in our cryptosystem.

	One possible way to break the security is by computing the $gcd$ of two different $N_i$s \cite{coron2011fully}. In other words, if $gcd(N_i, N_j)>1$, then the adversary can obtain one of the primes $k_i$ or $p_i$. Therefore, the adversary can break the security of the cryptosystem. This could be an issue if the same prime is used more than once to generate $N_i$s. To solve this issue, the recipients choose their primes randomly from a sufficiently large space, and these primes should be truly unpredictable numbers. In the proposed scheme, we assume that there is little to no chance of prime collision, where the probability that two recipients pick the same prime is extremely low \cite{mansour2019amoun,mansour2017multi,lindell2014introduction}. Therefore, the chance of primes collision is profoundly low.
	
	The element $e_i$ is generated using primes $k_i$ and $q_i$ in addition to the modular inverse $y_i^{-1}$, which is a noise parameter added to guarantee $gcd(N_i,e_i)=1$ \cite{coron2011fully,gentry2009fully}, as shown in equation \eqref{E5}. In the absence of $y_i^{-1}$, the element $k_i$ would be a divisor of both $e_i$ and $N_i$, which would imply that $k_i|gcd(e_i,N_i)$; therefore, the adversary can break the system. 
	
	\begin{definition} $\textbf{[gcd attack on the noise parameter]}$ For every probabilistic polynomial time adversary $\mathcal{A}$, there is a negligible function $\epsilon$ such that,
		\begin{center}
			Pr$[\mathcal{A}(N_i, e_i)=k_i]\le \epsilon(n)$,
		\end{center}
		
		where $N_i = k_i \cdot p_i$ and $e_i = k_i \cdot q_i+y^{-1}_i$ such that $k_i$, $p_i$ and $q_i$ are random primes with size $n-bits$, and $|y^{-1}_i|<2^\alpha$.
	\end{definition}
	
	According to \cite{coron2011fully}, the best-known $gcd$ attack to find $gcd(N_i, e_i)=k_i$, using fast multiplication, is with the asymptotic complexity of $2^\alpha \cdot O(n)$, for integers of size $n-bits$. Hence, $\epsilon(n)=1/(2^\alpha \cdot O(n))$. This implies the security of Theorem \ref{MainVIT}-(b) and proves the security of the parameter $e_i$. 	
	
	The last element of $K^+_i$, $d_i$, is generated using the two primes $k_i$ and $v_i$, as shown in equation \eqref{E_d}. Since the discrete logarithm problem is known to be mathematically and computationally infeasible, then getting $k_i$ and $v_i$ given $d_i$ is not possible, as the parameter $d_i$ is non-invertible because the inverse of $d_i$ is not unique \cite{diffie1976new,lindell2014introduction}.
	
	\begin{lemma}
		Let $\mathbb{G}$ be a finite group. Let $k_i, v_i$ be two prime numbers in the group $\mathbb{G}$. The element $N_i$ is the product of two primes $k_i$ and $p_i$. Setting $d_i \equiv v_i^{k_i}~(\bmod~N_i)$ gives the same distribution for $\hat{d_i}$ as choosing random $\hat{d_i} \in \mathbb{G}$, i.e., for any $\hat{d_i} \in \mathbb{G}$
		\begin{center}
			Pr[$v_i^{k_i}~(\bmod~N_i) \equiv \hat{d_i}$] = 1/$|\mathbb{G}|$.
		\end{center}
	\end{lemma}
	
	\begin{proof}
		Let $\hat{d_i} \in \mathbb{G}$ be arbitrary. Then,
		\begin{center}
			Pr[$v_i^{k_i}~(\bmod~N_i) = \hat{d_i}$] = Pr[$k_i = log_{v_i}{(d_i + q \cdot N_i)}$],
		\end{center}
		where $q$ represents the infinite quotients of $d_i~(\bmod~N_i)$.
		
		Since both $k_i$ and $v_i$ are unknown primes that are truly randomly generated, then the probability that $k_i$ is equal to $log_{v_i}{(d_i+ q \cdot N_i)}$ is exactly 1/$|\mathbb{G}|$ \cite{lindell2014introduction}.
	\end{proof}
	The above complexity implies the security of Theorem \ref{MainVIT}-(c) and proves the security of $d_i$ in our cryptosystem. 
	
	Next we show that the private key elements $k_i, v_i, y_i$ are secure against $gcd$ attack, by knowing the public key elements $N_i, e_i, d_i$. We consider the $gcd$ of every pair of public key elements as follows:
	
	\begin{itemize}
		\item \textbf{Case 1:} Let \textbf{$gcd(N_i,e_i)=g$}. If $g=k_i$, then since $g|e_i$ and $g|k_i\cdot q_i$, we get that $g|(e_i - k_i \cdot q_i) =y^{-1}_i$ which is impossible since $k_i,y_i$ are co-prime because $k_i$ is prime and $y_i < k_i$. Using the same reasoning as above $g\neq q_i$. If $g = y^{-1}_i$, then $g | (e_i - y^{-1}_i)=k_i \cdot q_i$. The latter statement does not hold true, as both $k_i, q_i$ are prime numbers and $y^{-1}_i < k_i,q_i$. The existence of the noise parameter $y^{-1}_i$ forces $gcd(N_i,e_i)=1$ \cite{coron2011fully,gentry2009fully}.
		
		\item \textbf{Case 2:} Let \textbf{$gcd(N_i,d_i)=g$}. If $g=k_i$, then because $d_i=v^{k_i}_i -q\cdot k_i\cdot p_i$, for some quotient $q$, then $g|(d_i + q\cdot k_i \cdot p_i)=v^{k_i}_i$, which is impossible as $g \nmid v_i$, since $v_i$ is prime. Using the same reasoning $g \neq p_i$. If $g=v_i$, then $v_i | (v^{k_i}_i -d_i)=q\cdot k_i \cdot p_i$. For this to happen, as both $k_i$ and $p_i$ are primes, we need to have $v_i=k_i$ or $v_i=p_i$, which is not the case as all $k_i, p_i, v_i$ are different primes.
		
		\item \textbf{Case 3:} Let \textbf{$gcd(e_i,d_i)=g$}. If $g = y^{-1}_i$, then $g | (e_i - y^{-1}_i)= k_i \cdot q_i$. For this to happen, as both $k_i,q_i$ are prime numbers, we need to have $y^{-1}_i=k_i$ or $y^{-1}_i=q_i$, which is not the case as $y^{-1}_i < k_i, q_i$. Following the same reasoning as in case 1, $g \neq k_i, q_i$. Furthermore, as in case 2, the inequalities $g\neq v_i, p_i$ hold true.\\

	\end{itemize}
	
	We continue by proving the security of the encryption elements, $N'_i$ and $e'_i$. In the initialization phase, the sender computes $N'_i$ for each recipient using parameters $N_i$ and $d_i$, in addition to the two randomly generated numbers $f_i$ and $t_i$, which are only known to the sender, as shown in equation \eqref{E_N}. For the adversary to be able to find $N'_i =N_i \cdot f_i + d_i \cdot t_i$ is by brute force attack only since there are infinitely many triples $(N'_i, f_i, t_i)$ that can satisfy the above equation. To enforce this statement, let $N'_i=z, f_i=x, t_i=y$ and the known values (coefficients) $N_i=a$ and $d_i=b$. Then, the equation $z=a\cdot x + b \cdot y$ represents a plane. This means that there are infinitely many points satisfying the equation for a given pair $(a,b)$, as shown in Lemma \ref{PlaneLemma}.

	\begin{lemma} \label{PlaneLemma}
		The complexity of finding a three dimensional point on a plane is $$f(x,y)=\lim_{x,y\to\infty} O(x \cdot y),$$
		where $x,y$ are the dimensions of a rectangle on the plane.
	\end{lemma}
	\begin{proof}
		$O(x \cdot y)$ is the number of steps needed to find all the points within a specific rectangular area of dimensions $x,y$ on the plane. As the rectangle expands to the whole plane, and therefore the number of points increases toward infinity, then the number of steps needed to find the specific point on the plane is $$f(x,y)=\lim_{x,y\to\infty} O(x\cdot y).$$ 
	\end{proof}

	For the sake of security of the proposed cryptosystem, it is crucial that the adversary can not find $N'_i$ and $e'_i$. To the contrary, since $gcd(e'_i,N_i)=1$, then the adversary could obtain $e_i^{'-1}~(\bmod~N_i)$ and retrieve any message $m_i$, after some mathematical operations. 
	
	Now, even in the case when the adversary finds $N'_i$, which is computationally infeasible, finding the corresponding $e'_i=e_i + N'_i \cdot r_i$, when given the public key parameters $N_i, e_i, d_i$, is still computationally infeasible. This is because the above equation represents a line $y=a\cdot x+b$, where $y=e'_i,x=r_i, a=N'_i$ and $b=e_i$, and the complexity of finding a point on a line is shown in Lemma \ref{XNiLemma}. Therefore, the adversary will not be able to find the modular inverse of $e'_i$ with respect to $N_i$. Thus, retrieving the message $m_i$ by computing $C\cdot e_i^{'-1}~(\bmod~N_i)$, is computationally infeasible. 
	
	\begin{lemma} \label{XNiLemma}
		The complexity of finding a specific point in a linear equation on two variables is $$g(z)=\lim_{z\to\infty} O(z),$$
		where $z$ is the number of integers within an interval of the line represented by the equation given.
	\end{lemma}
	\begin{proof}
		$O(z)$ is the number of steps needed to find all the integers within a specific interval of the line. As the interval expands to the whole line, and therefore the number of integers increases toward infinity, then the number of steps needed to find the specific point on the line is $$g(z)=\lim_{z\to\infty} O(z).$$ 
	\end{proof}
	
	Since there are additional random numbers $f_i, t_i ,r_i$ not known by the adversary in any case, then the adversary will not be able to compute the encryption element $S_i$ of any recipient $i$. This is because $S_i$ is the multiplication of $e'_i$ and $AX_i$, as shown in equation \eqref{E8N}.

	\begin{theorem} \label{ThAll}
		\textit{Let $f(x,y)$ be the function computing the complexity of finding the correct three dimensional point $(N'_i, f_i, t_i)$ on a given plane. Also, let $g(z)$ be the function computing the complexity of finding the correct two dimensional point $(e'_i,r_i)$ on a given line. Then, the complexity to find the correct $S_i$s is $O(n \cdot f(x,y) \cdot g(z))$.}	
	\end{theorem}
	\begin{proof}
		The proof for the theorem follows directly from Lemma \ref{PlaneLemma} and Lemma \ref{XNiLemma}. 
	\end{proof}
	The above complexity implies the security of Theorem \ref{MainVIT}-(d) and proves the security of $N'_i$ and $e'_i$ in our cryptosystem. 
	
	Next we show that our cryptosystem is indistinguishable under chosen plaintext attack $(IND-CPA)$. As mentioned in Section \ref{Enc_sec}, in every ciphertext, each sub-ciphertext $c_i=m_i \cdot S_i$ contains $r_i$, where $r_i$ is a random number generated for each encryption. Due to the randomization property of $r_i$, which causes nonlinearity in equation \eqref{E28}, for the same message vector $M$ and the same public keys, different ciphertexts are generated for different encryptions.
	
	Indeed, let $M$ be a message vector of size $n$ to be encrypted using $K_1^+, ..., K_n^+$ and $S_i=e'_i \cdot A_i \cdot \frac{X}{N_i}, \forall i \in \{1,2,\ldots,n\}$. Based on equation \eqref{E28}, if we encrypt $M$ two times we will get the two ciphertexts $C$ and $C"$ as follows:
		\texttt{\\}
				\texttt{\\}
						\texttt{\\}
	\begin{dmath}
		C=(m_1\cdot e'_1 \cdot A_1 \cdot \frac{X}{N_1}+...+m_i\cdot e'_i \cdot A_i \cdot \frac{X}{N_i}+...+m_n\cdot e'_n \cdot A_n \cdot \frac{X}{N_n})~(\bmod~X),
	\end{dmath}
		\texttt{\\}
	
	\begin{dmath}
		C"=(m_1\cdot e_1" \cdot A_1 \cdot \frac{X}{N_1}+...+m_i\cdot e_i" \cdot A_i \cdot \frac{X}{N_i}+...+m_n\cdot e_n" \cdot A_n \cdot \frac{X}{N_n})~(\bmod~X).
	\end{dmath}
	\texttt{\\}
	\texttt{\\}

	\begin{claim} \label{claimRandom}
		$C\neq C"$
	\end{claim}

	\begin{proof}
		For the sake of contradiction, assume that $C=C"$. After subtracting the above equations on both sides, we get that 
		\begin{dmath}
			C-C"=(m_1\cdot (e'_1-e_1") \cdot A_1 \cdot \frac{X}{N_1}+...+m_i\cdot (e'_i-e_i") \cdot A_i \cdot \frac{X}{N_i}+...+m_n\cdot (e'_n-e_n") \cdot A_n \cdot \frac{X}{N_n})~(\bmod~X).
		\end{dmath}
		\texttt{\\}
		Since $C=C"$, then
		\begin{dmath}\label{beforemodN}
			m_1\cdot (e'_1-e_1") \cdot A_1 \cdot \frac{X}{N_1}+...+m_i\cdot (e'_i-e_i") \cdot A_i \cdot \frac{X}{N_i}+...+m_n\cdot (e'_n-e_n") \cdot A_n \cdot \frac{X}{N_n} \equiv 0~(\bmod~X).
		\end{dmath}
		\texttt{\\}
		Now taking $\bmod~N_i$ on both sides of equation \eqref{beforemodN}, we obtain the following:
		\begin{dmath}\label{aftermodN}
			m_i\cdot (e'_i-e_i") \cdot A_i \cdot \frac{X}{N_i} \equiv 0~(\bmod~N_i).
		\end{dmath}
		
		Based on equation \eqref{E7}, $A_i$ is the modular inverse of $\frac{X}{N_i}$ with respect to $N_i$. Therefore, equation \eqref{aftermodN} reduces to
		
		\begin{dmath}\label{reduction}
			m_i\cdot (e'_i-e_i") \equiv 0~(\bmod~N_i).
		\end{dmath}
		
		After substituting $e'_i, e_i"$ by their equivalent expressions as in equation \eqref{E_e}, the equivalent of the latter equation will be,
		
		\begin{dmath}\label{substitution}
			m_i\cdot N'_i \cdot (r_i-r_i")   \equiv 0~(\bmod~N_i).
		\end{dmath}
		
		Replacing $N'_i=N_i \cdot r_i +d_i \cdot t_i$ and after applying $\bmod~N_i$ on both sides of equation \eqref{substitution}, we get
		
		\begin{dmath}\label{substitution2}
			m_i\cdot d_i \cdot t_i \cdot (r_i-r_i") \equiv 0~(\bmod~N_i).
		\end{dmath}
		
		Since $m_i, t_i < k_i$ and $k_i$ is a prime number, then the inverses of $m_i$ and $t_i$ with respect to $k_i$ exist. By multiplying equation \eqref{substitution2} on both sides by $m^{-1}_i \bmod~k_i$ and $t^{-1}_i \bmod~k_i$ we obtain the following equation,
		
		\begin{dmath}\label{substitution3}
			v_i \cdot (r_i-r_i") \equiv 0~(\bmod~k_i).
		\end{dmath}
		
		Following the same reasoning as above since $v_i<k_i$, multiplying by $v^{-1}_i \bmod~k_i$, the following holds:
		
		\begin{equation}\label{equu} 
		r_i-r_i" \equiv 0~(\bmod~k_i).
		\end{equation}
		
		As both $r_i, r_i"<k_i$, then $(r_i-r_i")<k_i$. Therefore, the only case when equation \eqref{equu} holds is when $r_i=r_i"$, which is impossible as both $r_i$ and $r_i"$ are truly randomly generated. 
		
	\end{proof}


	Let $N_{rand}=\{N_1", \ldots, N_n"\}$ be a randomly generated set of cardinality $n$, where $n$ is the number of recipients in the receiving group, and let $X"=\prod_{i=1}^{n} N_i"$. Also, let $E_{rand}=\{e_1", \ldots, e_n"\}$ be a randomly generated set of cardinality $n$, for each $i$. On the other hand, let $N_{real}=\{N'_1, \ldots, N'_n\}$ and $E_{real}=\{e'_1, \ldots, e'_n\}$ be two sets generated using $K^+_1, \ldots, K^+_n$, respectively. Moreover, let $M$ be a message vector of size $n$ to be encrypted.

	\texttt{\\}
	
	Now, consider the following two distributions:
	\begin{itemize}
		\item
		Distribution of the set $S_{real}=\{S_1, \ldots, S_n\}$ which is generated using the two real sets $N_{real}$ and $E_{real}$.\\
		\item 
		Distribution of the set $S_{rand}=\{S_1", \ldots, S_n"\}$ which is generated using the two random sets $N_{rand}$ and $E_{rand}$.
	\end{itemize}
	
	\texttt{\\}
	
	The adversary $\mathcal{A}$ solves the decision problem by effectively distinguishing two distributions, $S_{rand}$ and $S_{real}$, in polynomial time. In other words, given a challenge ciphertext $C$ from the encryption oracle, the adversary should know which distribution has been used in the encryption to generate $C$. Furthermore, the encryption oracle randomly selected a distribution $S_b$, where $b \in \{random, real\}$, and $\mathcal{A}$ should guess which distribution was chosen to get $C$.
	
	The advantage of the adversary to link the challenge ciphertext to the correct distribution should be negligible. Moreover, if finding the correct pair $(N'_i, e'_i)$ for all recipients $n$ is hard, then there is no significant change in the behavior of the adversary $\mathcal{A}$ when the distribution $S_{real}$ is replaced by the random distribution $S_{rand}$. Therefore, if we perform this substitution between the two distributions $S_{real}$ and $S_{rand}$, and the advantage of $\mathcal{A}$ is negligible, then the message vector $M$ is completely hidden, which implies security of the cryptosystem $AMOUN$.
	
	\begin{claim} \label{claimDistn}
		Pr[$C = E_{S_{real}}(M)$] = Pr[$C = E_{S_{rand}}(M)$]+$\epsilon$.
	\end{claim}
	\begin{proof} We will start by showing that if the input of the security parameters comes from $S_{real}$, as shown in equation \eqref{eqreal}, the adversary will have a negligible advantage to get information that helps in guessing the hidden bit $b$, given the challenge ciphertext $C$ and $K^+_1, \ldots, K^+_n$.
		
		\begin{dmath} \label{eqreal} 
			C=(m_1\cdot S_1+...+m_i\cdot S_i+...+m_n\cdot S_n)~(\bmod~X).
		\end{dmath}
		
		As mentioned above, it is computationally infeasible for the adversary to get $N'_i$ given $N_i$ and $d_i$, as equation \eqref{E_N} becomes linear function of the form $z=a\cdot x+b\cdot y$; therefore, there are infinitely many ordered pairs $(N'_i, f_i, t_i)$ that satisfy the linear function, as shown in Lemma \ref{PlaneLemma}. Also, even in the case when the adversary computes the proper $N'_i$, it is computationally infeasible for the adversary to get $e'_i$ given $e_i$ and $N'_i$, as equation \eqref{E_e} becomes linear function of the form $y=a\cdot x+b$; therefore, there are infinitely many ordered pairs $(e'_i, r_i)$ that could be the correct solution of the linear function as shown in Lemma \ref{XNiLemma}.

		On the other hand, if the input of the security parameters comes from $S_{rand}$, as shown in equation \eqref{eqrand}, then the adversary will have a negligible advantage to determine hidden bit $b$, given the challenge ciphertext $C$ and $K^+_1, \ldots, K^+_n$.
		
		\begin{dmath} \label{eqrand} 
			C=(m_1\cdot S_1"+...+m_i\cdot S_i"+...+m_n\cdot S_n")~(\bmod~X").
		\end{dmath}
		
		Since $\mathcal{A}$ cannot find a closed form solution for the nonlinear equation \eqref{eqrand} with $3 \cdot n$ unknown variables, then $\mathcal{A}$ will not be able to find $f_i", t_i", r_i"$, for any recipient $i$, given $C$. This implies that the distribution of the bit $b$ is independent from the view of $\mathcal{A}$; therefore, the advantage of $\mathcal{A}$ is negligible.
		
		Based on the above reasoning, from the view of $\mathcal{A}$, the encryption elements of both $S_{real}$ and $S_{rand}$ are random elements. Therefore, $\mathcal{A}$ cannot effectively distinguish two distributions in polynomial time. Thus, the advantage of $\mathcal{A}$ is negligible if $S_{real}$ is replaced by $S_{rand}$.
	\end{proof}

	
	Now, let $M_b \in MsgSp(v)$, $b \in \{0,1\}$, be two message vectors of size $n$ to be encrypted, where $|M_0|=|M_1|$. The cryptosystem is not $IND-CPA$ if the adversary $\mathcal{A}$ effectively distinguishes $M_0$ and $M_1$ in polynomial time, given a challenge ciphertext $C$. After the encryption oracle encrypts one of the message vectors at random, the encryption oracle provides $\mathcal{A}$ with $C$. Then the adversary $\mathcal{A}$ should guess which of the given message vectors relates to $C$.
	
	The advantage of $\mathcal{A}$ to link the challenge ciphertext to the correct message vector should be negligible. Moreover, since the coins $\{r_1,\ldots,r_n\}$ randomize the ciphertext $C$, and computing the security parameters $S_1, ..., S_n$ is hard, then the advantage of $\mathcal{A}$ is negligible, which implies the cryptosystem $AMOUN$ to be $IND-CPA$.

	\begin{claim}
		Pr[$C=E(M_0)$]=Pr[$C=E(M_1)$]$+\epsilon$.
	\end{claim}
	\begin{proof}
		Since the adversary knows the message vectors and the challenge ciphertext $C$, then in the equation \eqref{E28}, the adversary will have $3 \cdot n$ unknown variables $f_i, t_i,r_i, \forall i \in \{1,2,\ldots,n\}$. In the best case, when the number of recipients is at the minimum $n=2$, there are still $6$ unknown variables which cannot be rewritten in terms of the known variables because of nonlinearity. It is impossible to find a closed-form solution having $1$ nonlinear equation on $6$ variables, unless the advantage of a very specific relationship between the unknown variables, where the equation factoring can be completely taken place.
		
		Following the exact same reasoning as in Claim \ref{claimRandom} and Claim \ref{claimDistn}, the coin $r_i, \forall i \in \{1,2,\ldots,n\}$ is randomly generated for each encryption and cannot be rewritten in terms of the known variables. Since $\mathcal{A}$ cannot compute $e_i^{'-1}~(\bmod~N_i)$ because finding all $N'_i$s and $e'_i$s is extremely hard based on Theorem \ref{ThAll}, then $\mathcal{A}$ will not be able to retrieve any message $m_i$. 
		
		Therefore, $\mathcal{A}$ cannot learn any other information from knowing the challenge ciphertext $C$ and $M_0, M_1$, given to the encryption oracle, which implies that the probability that the encryption oracle encrypted, $M_0$, is the same as probability that the encryption oracle encrypted $M_1$, up to a negligible advantage $\epsilon=1/(O(n \cdot f(x,y) \cdot g(z)))$ obtained by Theorem \ref{ThAll}.
		
	\end{proof}
	The above three claims imply the indistinguishability of our proposed cryptosystem.
\end{proof}

$AMOUN$ is not secure against any form of chosen ciphertext attack $(CCA)$. Though no mathematical solution can be done to solve this issue, we propose padding as a possible method of making our scheme $CCA$ secure. As for $CCA$, it has not been solved for RSA without involving padding \cite{mansour2017multi,schneier2007applied,ferguson2011cryptography}; therefore, we cannot see an easy modification to allow for $CCA$ in our scheme as it stands even with modifications to the decryption.


\section{Time Complexity Analysis} \label{sec:TimeComplexity}

This section explains the time complexity of $AMOUN$ and compares it with RSA and Multi-RSA. Specifically, we explore the time complexity of initialization, encryption, and decryption. Then, we compare it with the time complexity of RSA and Multi-RSA, since they have a similar mathematical formulation as $AMOUN$. To calculate time complexity, we make certain assumptions. First, basic arithmetic operations such as addition and subtraction have time complexity $O(\floor*{log(n)}+1)$, where $n$ is the size of the largest decimal operand. Second, multiplication, division, and modulus have time complexity $O((\floor*{log(n)}+1)^2)$, where $n$ is the size of the largest decimal operand. Third, the exponentiation takes $O(\floor*{log(n)} \cdot (\floor*{log(w)}+1)^2)$ when using Binary Method, also known as the square and multiply method \cite{gordon1998survey}, where $n$ is the size of the decimal power and $w$ is the size of the decimal modulus. Fourth, for very large numbers, the extended Euclidean algorithm takes $O(\floor*{log(n)}^2)$, where $n$ is the size of the largest decimal operand \cite{lindell2014introduction}. Table \ref{tbl:Complexity} presents the comparison of time complexity for $AMOUN$, Multi-RSA, and RSA.


\begin{table*}[!t]
	\centering
	\caption{Time complexity analysis of $AMOUN$, Multi-RSA, and RSA}
	\label{tbl:Complexity}
	
	\begin{tabular}{llll}
		\toprule 
		& \textbf{Initialization} & \textbf{Encryption} &\textbf{Decryption for one recipient} \\ \midrule 
		\textbf{AMOUN}         &$O(n \cdot( (\floor*{log(X)}+1)^2 +\floor*{log(X)}^2))$ &$O(n \cdot ((\floor*{log(X)}+1)^2+ (\floor*{log(X)}+1)))$ &$O((\floor*{log(C)}+1)^2)$  \\ 
		\textbf{Multi-RSA}     &$O(n \cdot( (\floor*{log(X)}+1)^2 +\floor*{log(X)}^2))$&$O(n \cdot( \floor*{log(e_i)} \cdot (\floor*{log(N_i)}+1)^2+ (\floor*{log(X)}+1)^2))$&$O( \floor*{log(d_i)} \cdot (\floor*{log(N_i)}+1)^2)$\\ 
		\textbf{RSA}           &NA&$O(n \cdot( \floor*{log(e_i)} \cdot (\floor*{log(N_i)}+1)^2))$&$O( \floor*{log(d_i)} \cdot (\floor*{log(N_i)}+1)^2)$  \\ 
		\bottomrule
	\end{tabular}	
\end{table*}


Based on these assumptions, initialization for $AMOUN$ takes $O(n \cdot (4 \cdot (\floor*{log(X)}+1)^2 +\floor*{log(X)}^2+ \floor*{log(N_i)}+1))$, where $X$ is the largest operand and $n$ is number of recipients. $AMOUN$'s initialization requires $two$ multiplications and $one$ addition to compute $N'_i$, $one$ multiplication of $N_i$ to get $X$, performing the extended Euclidean algorithm to find $A_i$, and $one$ multiplication to get $AX_i$. Now as $X$ grows, the cost of other factors becomes negligible. Therefore, time complexity of $AMOUN$'s initialization is $O(n \cdot( (\floor*{log(X)}+1)^2 +\floor*{log(X)}^2))$. Moreover, Multi-RSA initialization takes $O(n \cdot( (\floor*{log(X)}+1)^2 +\floor*{log(X)}^2))$ since it performs two multiplications and the Extended Euclidean Algorithm for $n$ recipients, where RSA scheme does not have initialization \cite{mansour2017multi}.

Encryption for $AMOUN$ takes $O(n \cdot (4 \cdot (\floor*{log(X)}+1)^2+ 2 \cdot (\floor*{log(X)}+1)))$, where $X$ is the largest operand and $n$ is number of recipients. $AMOUN$'s encryption requires $one$ multiplication and $one$ addition to compute $e'_i$, $one$ multiplication to get $S_i$, $one$ multiplication and $one$ addition to find the sub-ciphertext $c_i$, and $one~\bmod~X$ operation. Now as $X$ grows, the constants become negligible. Therefore, for $AMOUN$, time complexity of encryption is $O(n \cdot ((\floor*{log(X)}+1)^2+ (\floor*{log(X)}+1)))$. On the other hand, Multi-RSA encryption takes $O(n \cdot( \floor*{log(e_i)} \cdot (\floor*{log(N_i)}+1)^2+ (\floor*{log(X)}+1)^2)+(\floor*{log(X)}+1)^2)$, where $e_i$ is the public key, $N_i$ is the modulus, $X$ is the product of all $N_i$s, and $n$ is number of recipients. Multi-RSA encryption requires $one$ exponentiation and $one$ multiplication, for each recipient $i$, and $one~\bmod~X$ operation at the end. Now as $X$ grows, the $\bmod~X$ operation become negligible. Therefore, the time complexity of Multi-RSA encryption is $O(n \cdot( \floor*{log(e_i)} \cdot (\floor*{log(N_i)}+1)^2+ (\floor*{log(X)}+1)^2))$ \cite{mansour2017multi}. RSA's encryption takes $O(\floor*{log(e_i)} \cdot (\floor*{log(N_i)}+1)^2)$ for one recipient, as it performs only $one$ exponentiation operation. Therefore, RSA takes $O(n \cdot( \floor*{log(e_i)} \cdot (\floor*{log(N_i)}+1)^2))$ for $n$ recipients, where $e_i$ is the public key and $N_i$ is the modulus.

Decryption for $AMOUN$ takes $O(3 \cdot (\floor*{log(C)}+1)^2)$ for one recipient, where $C$ is the ciphertext. We take $mod~k_i$, then multiply the result by the random number $y_i$, and we take $mod~v_i$ to get $m_i$. Now as $C$ grows, the constant becomes negligible. Therefore, $AMOUN$ has a time complexity of $O((\floor*{log(C)}+1)^2)$. On the other hand, decryption of both RSA and Multi-RSA takes $O( \floor*{log(d_i)} \cdot (\floor*{log(N_i)}+1)^2)$ for $one$ recipients, where $d_i$ is the private key and $N_i$ is the modulus.

\section{Evaluation}\label{sec:Evaluation}

This section presents the performance evaluation of $AMOUN$ and shows its comparative analysis with respect to RSA and Multi-RSA since both of them rely on the properties of modulus \cite{ding1996chinese} and prime factorization problem \cite{rivest1978method}. The key size was varied from $1024-bits$ to $6144-bits$. All experiments were performed on an $Intel~Core~i7-3517U~CPU$ with $8GB$ memory. $AMOUN$'s core library, RSA, and Multi-RSA were implemented in $Visual~Studio~2017~v15.3$ with $.NET~4.7$ using $C\#~7.1$. We ran each test $1000$ times, and the results were averaged to remove any outliers.

Keys of RSA and Multi-RSA were generated from primes $P_i$ and $Q_i$. The recurrence of prime numbers was avoided during any single test. The size of both $P_i$ and $Q_i$ is $1024-bits$, the size of $e_i$ and $d_i$ is also roughly $1024-bits$, and the size of $N_i$ is close to $2048-bits$. Public key $e_i$ was randomly generated such that it is less than $P_i$ and co-prime with $(P_i-1)\cdot(Q_i-1)$. Private key $d_i$ was found as the modular inverse of $e_i$ with respect to $(P_i-1)\cdot(Q_i-1)$. On the other hand, keys of $AMOUN$ were generated from primes $k_i$, $p_i$, $q_i$, and $v_i$. Size of $N_i$, $e_i$ and $d_i$ is close to $2048-bits$ for primes of size $1024-bits$, $4096-bits$ for primes of size $2048-bits$, and $6144-bits$ for primes of size $3072-bits$, respectively. Lastly, messages were randomly generated such that they were less than every prime we use. 

Performance evaluation of $AMOUN$'s initialization, encryption, and decryption was done compared to the corresponding phases of Multi-RSA. Since RSA does not have an initialization phase, the comparison was done with respect to encryption and decryption only. The first phase is initializing the needed parameters in the encryption phase using public keys $K^+_1, ..., K^+_n$ accordingly for $n$ recipients. The second phase is encrypting $n$ messages, $m_1, ..., m_n$, using the initialized parameters, accordingly. After that, producing ciphertext $C$ for $AMOUN$ and Multi-RSA, where concatenating every sub-ciphertext $c_i$ to come up with ciphertext $C$ for RSA. This gives us a fair comparison to $AMOUN$ because RSA is not a multi-recipient cryptographic scheme. For the third phase, we decrypt ciphertext $C$ using private key $K_i^-$ to recover the original message $m_i$.

\subsection{Initialization: $AMOUN$ and Multi-RSA} \label{sec:Ini_Results}

We compare the total initialization computational cost of $AMOUN$ to that of Multi-RSA, as shown in Fig.~\ref{fig:AMOUNInitial}. Initialization time for both $AMOUN$ and Multi-RSA includes initializing the needed parameters in the encryption phase for a receiving group of size $n$ in order to encrypt the messages using public keys $K^+_1, ..., K^+_n$, accordingly. The size of public keys that $AMOUN$ uses is $2048-bits$, whereas Multi-RSA uses public keys of size $1024-bits$.

From Fig.~\ref{fig:AMOUNInitial}, it is apparent that Multi-RSA has better performance than $AMOUN$, when the number of recipients is up to $10$. Moreover, Multi-RSA shows $3\%$ lower average processing time than $AMOUN$. The reason is that, for $AMOUN$, the parameter $N'_i$ is computed for every recipient in the initialization phase instead of the encryption phase, as shown in equation \eqref{E_N}. This small overhead improves the performance of the encryption phase, where the initialization phase is only done once for the group before the start of the encryption phase. On the other hand, there is no initialization phase for RSA scheme.

\begin{figure}
	\includegraphics[width=\columnwidth]{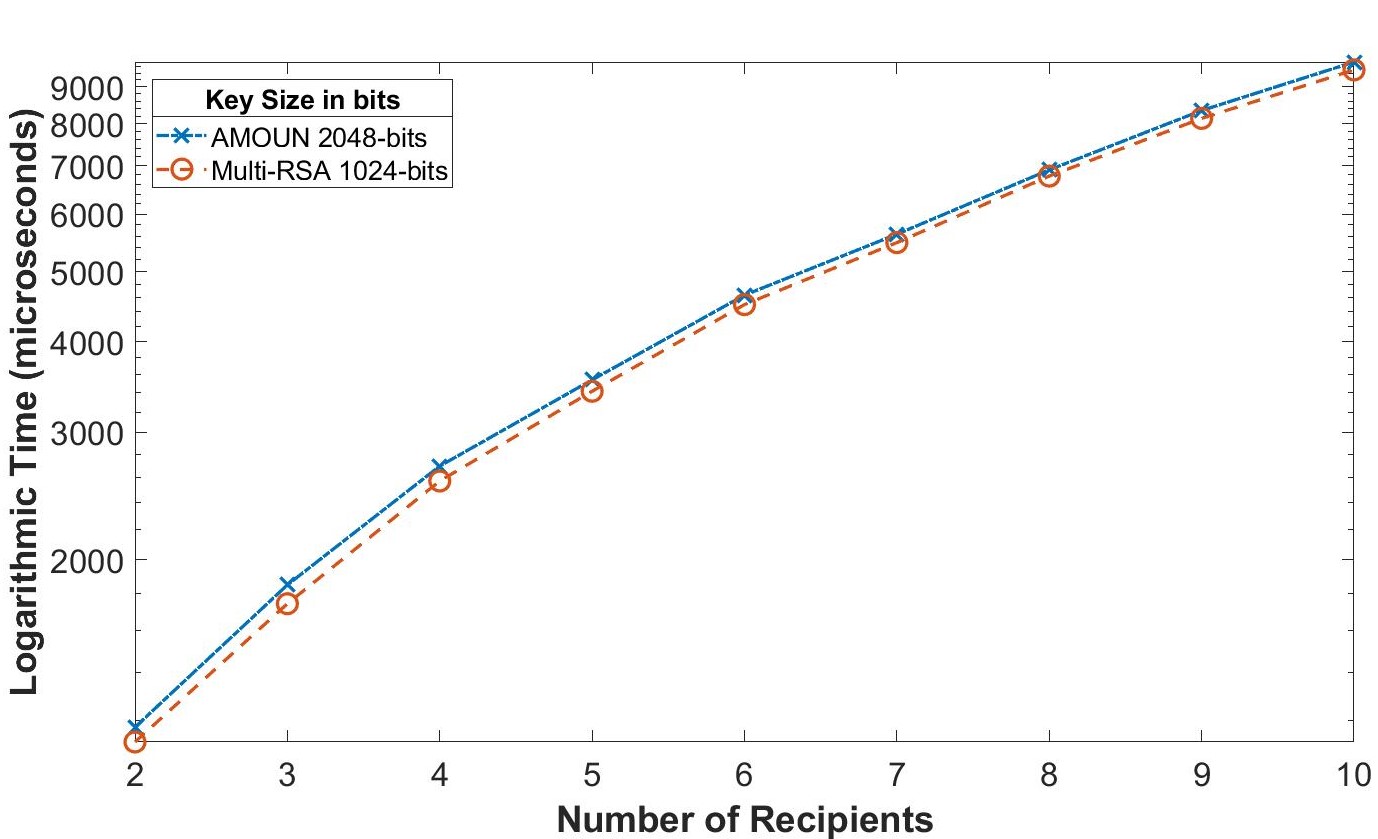}
	\caption{Total initialization time: $AMOUN$ with key size 2048-bits and Multi-RSA with key size 1024-bits}
	\label{fig:AMOUNInitial}
\end{figure}

\subsection{Encryption: $AMOUN$, RSA , and Multi-RSA} 

For encryption phase, we encrypt $n$ randomly generated messages with the same size using $K^+_1, ..., K^+_n$ for both RSA and Multi-RSA, and $S_1, ..., S_n$ for $AMOUN$, accordingly, then compute $C$ for $AMOUN$ and Multi-RSA schemes. On the other hand, to perform a fair comparison with RSA, we concatenate every sub-ciphertext $c_i$ to get $C$ for RSA. This is because RSA is not a native multi-recipient cryptographic scheme.

Fig.~\ref{fig:AMOUNEnc} compares the computational cost of the encryption algorithm of $AMOUN$ with key sizes $2048,~4096,~6144-bits$ to that of RSA and Multi-RSA with key size $1024-bits$. It is clear that, with respect to time required for encryption, $AMOUN$ performs significantly better than both RSA and Multi-RSA, even for larger key sizes when a number of recipients is up to $10$. Moreover, it can be seen that when the number of recipients increases, the performance gain of $AMOUN$, in terms of computational cost for encryption, will be further improved. Hence, $AMOUN$ is particularly more useful for applications having a high number of recipients. Quantitative analysis reveals that $AMOUN$ with $2048-bits$ key size shows $99\%$ lower average computational cost than both RSA and Multi-RSA with $1024-bits$ key size. Also, $AMOUN$ with $4096-bits$ key size shows $98\%$ and $99\%$ lower average computational cost than RSA and Multi-RSA with $1024-bits$ key size, respectively. Moreover, $AMOUN$ with $6144-bits$ key size shows $97\%$ and $98\%$ lower average computational cost than RSA and Multi-RSA with $1024-bits$ key size, respectively.

\begin{figure}
	\includegraphics[width=\columnwidth]{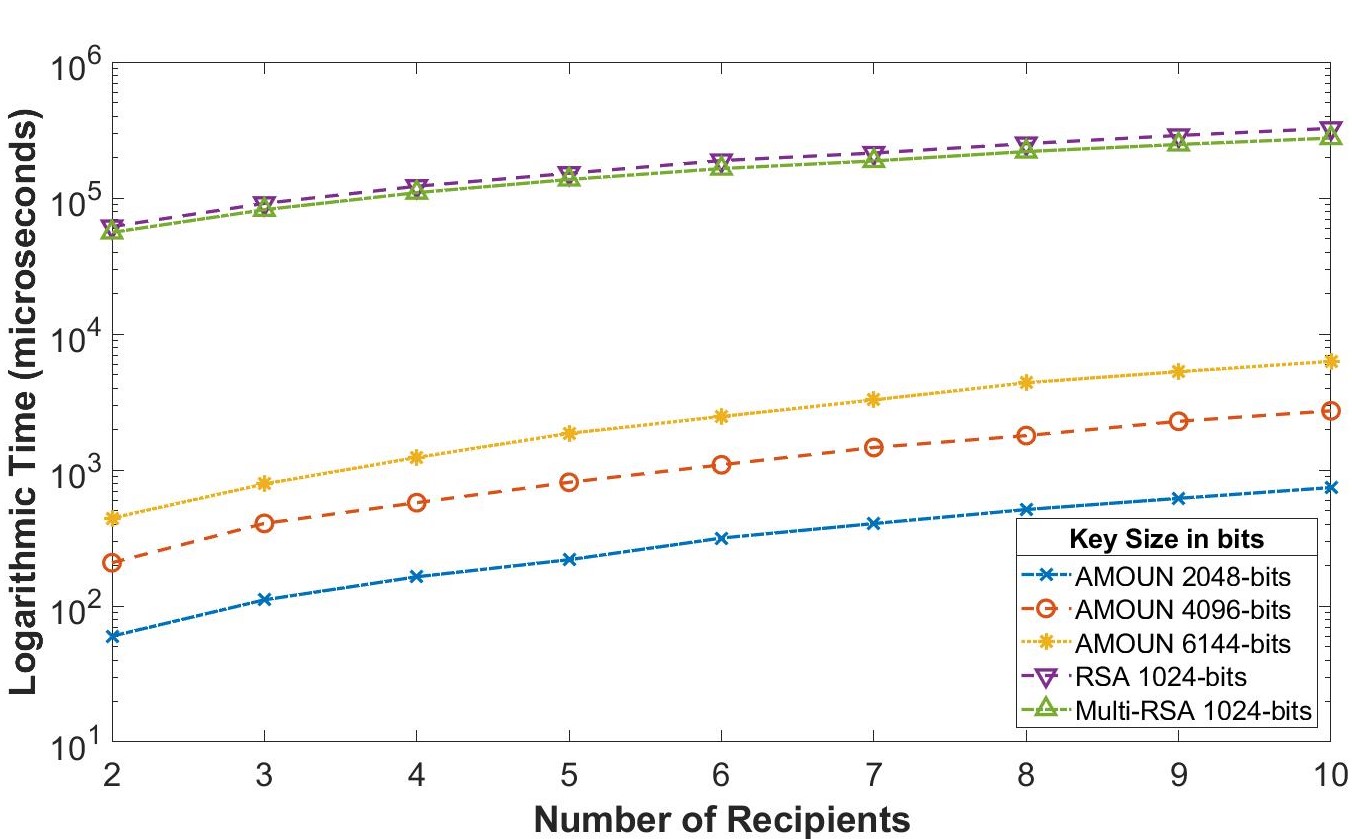}
	\caption{Total encryption time: $AMOUN$ with key sizes 2048, 4096, and 6144-bits, RSA and Multi-RSA with key size 1024-bit}
	\label{fig:AMOUNEnc}
\end{figure}

\subsection{Decryption: $AMOUN$, RSA, and Multi-RSA}

In order to compare $AMOUN$ with RSA and Multi-RSA in terms of time to perform decryption of ciphertexts, we compute the decryption time for every recipient for all schemes. In other words, we compute the time all recipients take to decrypt $C$ using their private keys $K_i^-$s to get all messages $m_i, \cdots, m_n$ for each scheme. Fig.~\ref{fig:AMOUNDecOne} compares the decryption computational cost of $AMOUN$ with key sizes $1024,~2048,~3072-bits$ to that of RSA and Multi-RSA with key size $1024-bits$. Fig.~\ref{fig:AMOUNDecOne} represents decryption time for RSA and Multi-RSA together since both use the same decryption algorithm. Note that the performance of $AMOUN$ is better than both RSA and Multi-RSA, as its decryption time is significantly less, even for larger key sizes when the number of recipients is up to $10$. Furthermore, $AMOUN$ with $1024-bits$ and $2048-bits$ key sizes shows $99\%$ lower average computational cost than RSA and Multi-RSA with $1024-bits$ key size. Also, $AMOUN$ with $3072-bits$ key size shows $98\%$ lower average computational cost than RSA and Multi-RSA with $1024-bits$ key size.

\begin{figure}
	\includegraphics[width=\columnwidth]{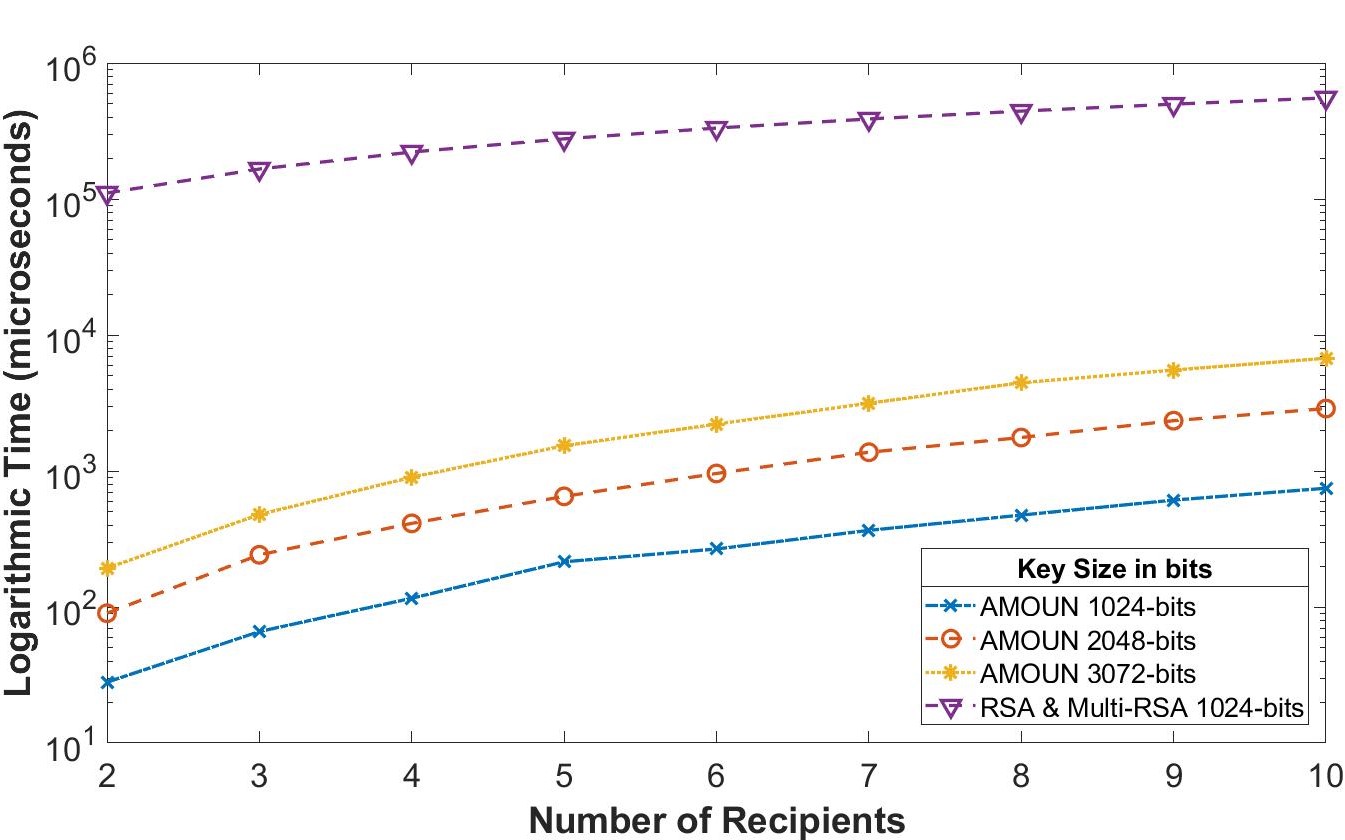}
	\caption{Total decryption time: $AMOUN$ with key sizes 1024, 2048, and 3072-bits, RSA and Multi-RSA with key size 1024-bit}
	\label{fig:AMOUNDecOne}
\end{figure}

\section{Discussion} \label{sec:Discussion}

In this section, we analyze the performance gain obtained by $AMOUN$ in the encryption phase. We perform a comparative analysis of our scheme with existing solutions such as RSA and Multi-RSA. We discuss the advantages of $AMOUN$ in terms of timeliness and scalability. This can be demonstrated by applying our scheme in real-time wireless connected environments having resource constraints such as WSNs and VANETs.

It can be concluded from section \ref{sec:Ini_Results} and Fig.~\ref{fig:AMOUNInitial} that there is a slight overhead of $AMOUN$'s initialization compared to Multi-RSA's. However, this initialization will occur only once per receiving group at the sender side, which is negligible considering the $99\%$ gain in terms of the computational cost of encryption. Fig.~\ref{fig:AMOUNIni+Enc} shows the significant improvement in performance, by decoupling initialization from encryption, and computing $N'_i$ in the initialization phase instead of the encryption phase, as shown in equation \eqref{E_N}. We obtain this performance gain in the encryption phase by reducing its complexity from $O(n \cdot( (\floor*{log(X)}+1)^2 +\floor*{log(X)}^2+(\floor*{log(X)}+1)))$ to $O(n \cdot ((\floor*{log(X)}+1)^2+ (\floor*{log(X)}+1)))$, where $n$ is the number of recipients, and $X$ is the multiplication of all $N_i$s of receiving group, as shown in equation \eqref{E6}. Quantitative analysis reveals that, with $2048-bits$ and $4096-bits$ keys size, the encryption is on average $93$ and $92$ times faster than the coupling of the initialization and encryption, respectively. It is clear from these results and Fig.~\ref{fig:AMOUNIni+Enc}, that we gain more improvement in the performance of encryption when the size of key increases.

\begin{figure}
	\includegraphics[width=\columnwidth]{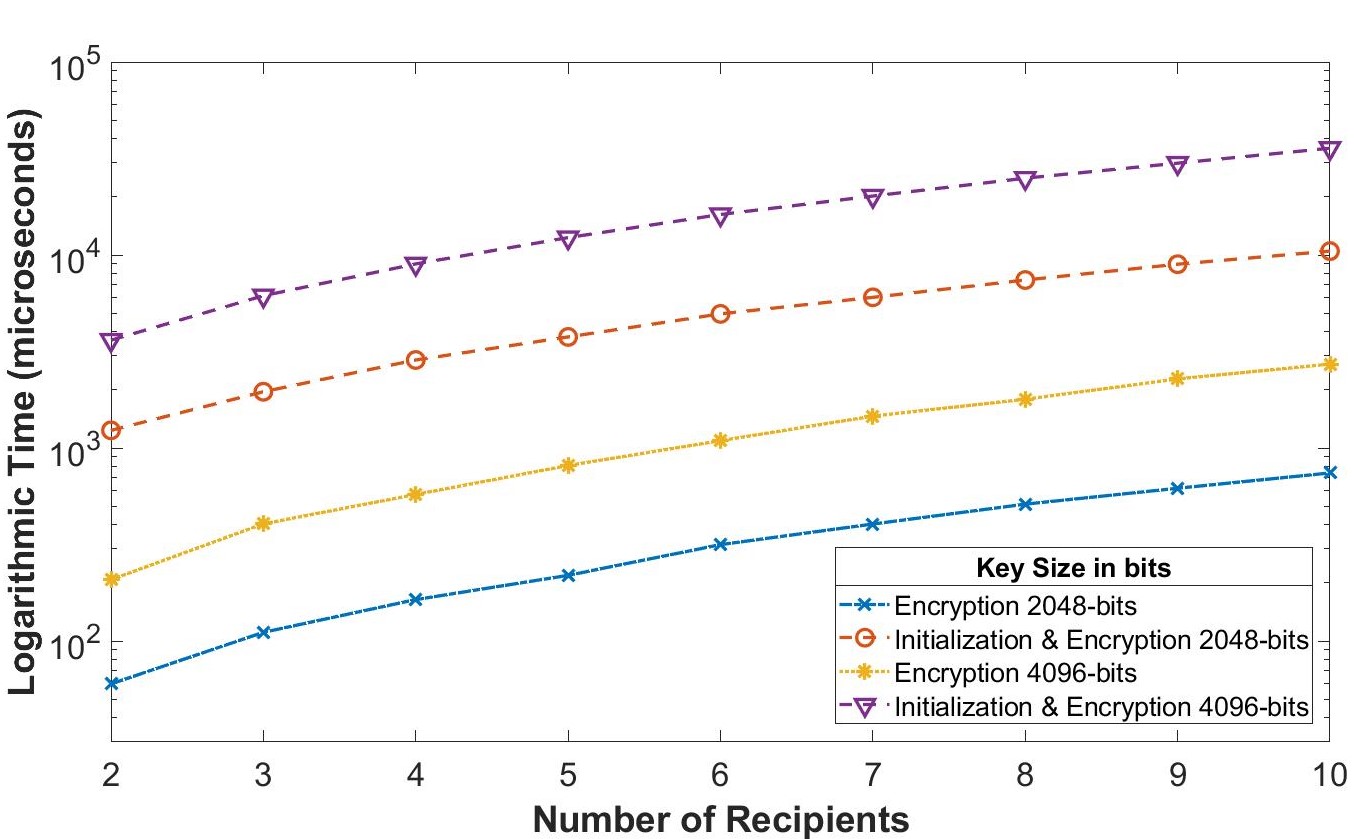}
	\caption{Total encryption time for $AMOUN$ with keys size 2048 and 4096-bits: encryption vs. initialization \& encryption.}
	\label{fig:AMOUNIni+Enc}
\end{figure}

As $AMOUN$ is a pure asymmetric multi-recipient cryptographic scheme, it overcomes the overhead of key distribution. On the other hand, for encryption and decryption, $AMOUN$ introduces a lightweight multi-recipient cryptographic scheme. For encryption and decryption, existing asymmetric solutions introduce huge computational cost compared to what $AMOUN$ exhibits. Note that the lightweight nature of $AMOUN$ does not affect its strength compared to other known multi-recipient asymmetric cryptographic schemes. For example, Fig.~\ref{fig:AMOUNEnc} shows that $AMOUN$ with $2048-bits$ key size is on average $538$ times faster than Multi-RSA with $1024-bits$ key size for encryption when the number of recipients is up to $10$. Moreover, Fig.~\ref{fig:AMOUNDecOne} shows that, for $1024-bits$ key size, $AMOUN$ is on average $1035$ times faster than Multi-RSA for decryption, when the number of recipients is up to $10$.

Note that for real time applications, when the sender has to perform secure multi-recipient communication and is resource deficient (less CPU, memory, and energy), RSA and Multi-RSA would become inapplicable due to their high computational cost. For example, our analysis shows that using an $Intel~Core~i7-3517U~CPU$ with $8GB$ memory, when no other applications were running, $AMOUN$ took $743~microseconds$ compared to Multi-RSA which consumed $325~milliseconds$, and took $747~microseconds$ to $553~milliseconds$, for encryption and decryption, respectively, when the size of the receiving group is $10$. From Fig.~\ref{fig:AMOUNEnc} and \ref{fig:AMOUNDecOne}, it is clear that there is a huge difference in terms of computational cost between $AMOUN$ and Multi-RSA, for encryption and decryption, respectively. In wireless mobile infrastructure-based networks, such as WSNs and VANETS, where the sender (i.e. RSU, AP, and BS) needs to securely communicate with a mobile receiving group, $AMOUN$ will be a preferable choice compared to RSA and Multi-RSA. Even for infrastructure-less based mobile networks, where the sender could be a mobile cluster head, the above conclusion holds true.

In order to quantitatively assess the reduction of communication overhead compared to traditional unicast-based schemes, assume the size of data needed to be sent to each recipient is $\lambda$, maximum allowable size for the packet is $\gamma$, and the number of recipients is $n$. $AMOUN$ introduces communication overhead of factor $\mu/\gamma$, where $\mu=\lambda\cdot n$ is the total encrypted ciphertext size. As $\lambda$ approaches to $\gamma$, which is the worst-case scenario for $AMOUN$, performance improvement over traditional schemes will be close to 1. In all other scenarios, $AMOUN$ will always involve less communication overhead compared to traditional schemes. Note traditional unicast-based schemes always requires at least $n$ number of transmissions. In addition, $AMOUN$ requires zero communication overhead among recipients or between sender and recipients when the recipient group changes. This makes $AMOUN$ ideal for systems, such as WSNs and VANETs, where there is little or no trust among individual recipients.

From Fig.~\ref{fig:AMOUNEnc}, it is apparent that as the size of receiving group increases, the computational cost for existing solutions, like Multi-RSA, grows exponentially which makes these schemes less scalable in wireless distributed network environments. Moreover, quantitative analysis reveals that, $AMOUN$ with $2048-bits$ key size took $743~microseconds$ compared to Multi-RSA which consumed $325~milliseconds$ with $1024-bits$ key size, $AMOUN$ took $14~milliseconds$ compared to Multi-RSA which consumed $1.5~seconds$, and took $127~milliseconds$ to $4.2~seconds$, when the size of the receiving group are $10$, $50$, and $150$, respectively. Based on these results, it is clear that as the size of the receiving group increases, $AMOUN$ is more scalable than Multi-RSA. Furthermore, for resource-constrained environments, the use of $AMOUN$ to secure multi-recipient communications is practical. One of our future works is to validate the applicability of $AMOUN$ in WSNs and VANETs.

\section{Conclusion} \label{sec:Conclusion}

This paper presented $AMOUN$, a novel cryptographic scheme. $AMOUN$ effectively integrates the mathematical formulations of CRT, prime factorization, discrete logarithm, and the use of the noise parameter, to achieve an efficient asymmetric multi-recipient cryptosystem. The proposed scheme overcomes many challenges other multi-recipient cryptographic schemes face, including but not limited to, the possibility of compromising group privacy, the collusion among recipients, and the need for key distribution. Security analysis shows that $AMOUN$ is indistinguishable under adaptive chosen plaintext attack. Moreover, complexity and performance analyses show the effectiveness of $AMOUN$. Empirical results show that $AMOUN$ introduces lower average computational cost, compared to both RSA and Multi-RSA, for both encryption and decryption. The performance of $AMOUN$ increases significantly compared to RSA and Multi-RSA when the size of both keys and multi-recipient group increase.



\section*{Acknowledgment}

This research work was partially supported by the National Science Foundation under Grant CNS-1815724. Any opinions, findings, and conclusions or recommendations expressed in this material are those of the authors and do not necessarily reflect the views of the National Science Foundation.

\bibliographystyle{IEEEtran}
\bibliography{Refe}

\end{document}